\documentclass[a4paper, UKenglish, cleveref, autoref, thm-restate, comments]{lipics-v2021}

\pdfoutput=1
\usepackage[T1]{fontenc}
\usepackage[utf8]{inputenc}

\usepackage{comments}
\usepackage{style}
\usepackage{proof_style}
\usepackage{tikzit}
\nolinenumbers

\pdfoutput=1 %
\title{Causality in Pure Quantum Computation with Quantum Control}

\author{
  Kengo Hirata
}{
  University of Edinburgh, United Kingdom
  \and Kyoto University, Japan
}{
  k.hirata@sms.ed.ac.uk
}{
  https://orcid.org/0009-0005-4416-2655
}{}

\author{
  Takeshi Tsukada
}{
  Chiba University, Japan
}{
  t.tsukada@acm.org
}{
  https://orcid.org/0000-0002-2824-8708
}{}

\authorrunning{K. Hirata and T. Tsukada} %

\Copyright{Kengo Hirata and Takeshi Tsukada} %



\category{} %

\funding{This work was supported by JST CREST, Japan, Grant Number JPMJCR25I5.}%

\EventEditors{Claudia Faggian and Joost-Pieter Katoen}
\EventNoEds{2}
\EventLongTitle{41st Annual Symposium on Logic in Computer Science (LICS 2026)}
\EventShortTitle{LICS 2026}
\EventAcronym{LICS}
\EventYear{2026}
\EventDate{July 20--23, 2026}
\EventLocation{Lisbon, Portugal}
\EventLogo{}
\SeriesVolume{380}
\ArticleNo{57}

\begin{document}
\maketitle

\begin{abstract}
Indefinite causal order is a characteristic phenomenon in quantum computation, with examples including the quantum SWITCH and the OCB process. Not all such processes are believed to be physically realizable: while some implementations of the quantum SWITCH have been proposed, the OCB process is suspected to be unrealizable. This difference in realizability is commonly attributed to constraints imposed by physical causality.

This paper studies such a causality issue in a higher-order setting, proposing a typed lambda calculus with quantum control and its categorical semantics.
Our calculus extends pure quantum computation with higher-order functions and quantum conditional branching, and it is equipped with a type system based on intuitionistic BV logic to enforce causality. We also present a novel model that is closely related to the Caus construction, by which we prove that some physically-unrealizable processes are not definable in our language.
 \end{abstract}

\section{Introduction}
\label{sec:intro}

Quantum computation is a computational paradigm that exploits the principles of quantum mechanics, and it can efficiently solve certain problems that are not efficiently solvable by any known classical algorithm.
Naturally, the scope of quantum computation is constrained by the operations that quantum mechanics allows.
In the standard formalism, a physical system is represented by a Hilbert space, and the physically admissible transformations between Hilbert spaces are well understood.
Concretely, depending on the class of operations that one allows, physical processes are modeled by unitary transformations, isometries, or quantum channels.
Furthermore, every physically realizable operation admits a description in terms of quantum circuits.

However, the situation changes dramatically once we move to second-order operators, \ie transformations that take quantum channels as inputs and return a new quantum channel.
Such transformations are called \emph{supermaps}.
As we shall see below, this setting exhibits novel and intriguing phenomena that do not arise at first order, but many fundamental questions remain open, including which transformations are physically realizable and how such operations can be described in an appropriate language.

One such phenomenon is \emph{indefinite causal order}, of which a prominent example is the \emph{quantum SWITCH}~\cite{Chiribella2013}.
The quantum SWITCH takes two quantum channels $A$ and $B$ as inputs and behaves as \( A \circ B \) or \( B \circ A \) depending on a control qubit.
In particular, when the control qubit is in a superposition of \( \ket{0} \) and \( \ket{1} \), the quantum switch performs a ``mixture'' or ``superposition'' of \( A \circ B \) and \( B \circ A \), hence ``which of A and B was executed first'' is not well-defined.
There is also research that investigates computational advantages that exploit this property~\cite{Araujo2014,Chiribella2021,Colnaghi2012,Ebler2018,Kristjansson2020,Kristjansson2024,switch-advantage2}.
Another instance of this phenomenon is the supermap introduced by Oreshkov, Costa and Brukner~\cite{Oreshkov2012-ocb}, which we call the \emph{OCB process}.
Its remarkable feature is the violation of the \emph{causal inequality}~\cite{Branciard_2016,Oreshkov2012-ocb,PurvesS21,Wechs2021}, which is unviolatable under classical causality.

Although the quantum SWITCH~\cite{Chiribella2013} and the OCB process~\cite{Oreshkov2012-ocb} are both interesting supermaps exhibiting indefinite causal order, their physical realizability is in different situations.
For the quantum SWITCH, several implementation proposals and experimental results have been reported~\cite{Friis2014,Liu2024,switch-experiment2,switch-experiment-battery}.
In contrast, to the best of our knowledge, no implementation or experimental realization of the OCB process has been given, and it is rather suspected to be unrealizable~\cite{AraujoFNB17,PurvesS21}.

The precise boundary between realizable and unrealizable supermaps is of significant interest, yet remains unclear.
To explore the boundary, or as candidates of the boundary, various classes of supermaps have been proposed.

One of the most fundamental is the class of \emph{pure supermaps}~\cite{AraujoBCFGB2015,AraujoFNB17}.
This is the class of supermaps that,
when the inputs are ``pure'' quantum operations,
\ie those involving neither measurement nor discarding,
similarly return a pure operation.
The purity distinguishes between the quantum SWITCH and the OCB process: the quantum SWITCH is pure, but the OCB process is not.
Several other classes have been proposed~\cite{FeixAB2016,OreshkovG2016,Wechs2021} as subclasses of \emph{purifiable supermaps},
\ie pure supermaps followed by measurements,
making purity a useful baseline notion.

Kissinger and Uijlen~\cite{Kissinger2019} and Simmons and Kissinger~\cite{Simmons2022,Simmons2024} studied supermaps from categorical and logical viewpoints.
Kissinger and Uijlen~\cite{Kissinger2019}
have introduced the \emph{Caus construction},
which yields a category that can model higher-order causal structure in probabilistic and quantum theories.
Simmons and Kissinger~\cite{Simmons2022,Simmons2024} showed that the resulting category provides a model for
causality-aware extensions of linear logic,
namely BV-logic~\cite{Guglielmi2007} and pomset-logic~\cite{retore_pomset_1997}.
This provides a striking link between notions of causal structure
in physics and in logic.
Hefford and Wilson~\cite{BV-model-HeffordWilson_2025} have extended their work and created a model which characterizes pure-reversible
quantum computation%
\footnote{Here, by pure-reversible, we mean the subclass of quantum computation
  where all first order functions are unitaries
  and higher order functions preserves them.}
and proved that their model can also interpret BV-logic,
strengthening the relation of causal structure in the two field.

The above research can be regarded as giving a declarative description of realizable supermaps.
The complementary direction is to ask what kinds of constructions are allowed to build realizable supermaps.
In terms of programming language, this amounts to asking what kinds of language constructs can be introduced to a programming language for supermaps.

\emph{Quantum control} is a construction that has recently attracted attention.
It is a programming principle for constructing pure higher-order processes~\cite{Wechs2021}, allowing a program to form coherent superpositions of different executions.
The quantum SWITCH is the canonical example of this idea, creating a superposition of different applications of input operations.
Importantly, this mechanism remains causally constrained: it can generate indefinite causal order, but by itself cannot violate causal inequalities~\cite{PurvesS21,Wechs2021}, unlike processes such as the OCB process.
Thus quantum control occupies an intermediate regime: it goes beyond fixed-order higher-order programs while still falling short of processes that violates causal inequalities.
This makes quantum control a natural test case for understanding causality in pure higher-order quantum computation.

The above-mentioned causal models capture causal constraints on higher-order processes, but they cannot interpret quantum control at least in an evident way.
This paper studies the causality in the coexistence of quantum control and higher-order functions.

\subparagraph{New Causality Issue in Higher-order Programs.}
The first contribution of this paper is to identify a new causality issue arising from the combination of quantum control and higher-order functions.
Quantum control, exemplified by the controlled-\( U \) gate (that applies \( U \) when the control qubit is \( \ket{1} \) and acts as the identity when the control qubit is \( \ket{0} \)), is a fundamental ingredient of quantum computation and also underlies the quantum SWITCH.
However, a quantum-controlled higher-order function gives rise to nontrivial causal problems.

We use the notation
$\qifx{M}{N_1}{N_2}$
for the syntax of quantum control in our language
to represent positively controlled $N_1$
and negatively controlled $N_2$
by a qubit $M$
in analogy with classical conditional $\cifx{M}{N_1}{N_2}$.
For example, the term $\qifx{x}{\gX\,y}{y}$
represents controlled-NOT gate applied to qubits $x$ and $y$,
where $\gX$ is the NOT gate that flips $\ket0$ to $\ket1$ and vice versa.
In contrast to the classical case,
$\qifx{x}{\gX\,y}{y}$ should have the tuple type $\qbit \otimes \qbit$
rather than just $\qbit$ since the control qubit should live after
the execution of the term.

A na\"ive generalization of this construct to possibly higher-order \( N_1 \) and \( N_2 \) would have the following typing rule:
\begin{proofrule}
  \infer*{
    \Gamma \vdash M \colon \qbit
    \\
    \Gamma' \vdash N_1 \colon A
    \\
    \Gamma' \vdash N_2 \colon A
  }{
    \Gamma, \Gamma' \vdash \qifx{M}{N_1}{N_2} \colon \qbit \otimes A
  }
  \tag{$\ast$}\label{rule:wrong-qif}
\end{proofrule}
This na\"ive typing rule for quantum control
allows some physically unrealizable program to be typed,
when \( A \) is higher-order.

\begin{figure}[t]
  \centering
  \begin{quantikz}[row sep = 1.2em]
    \lstick{\( x \)} & \ctrl{1} & \rstick{$x'$} \qw \\
    \lstick{$f_\mathrm{in}$} & \targ{} & \rstick{$f_\mathrm{out}$} \qw
  \end{quantikz}
  \quad$\Rightarrow$\hspace*{-1.5em}
  \begin{quantikz}[row sep = 1.2em]
    \lstick{\( x \)} & \ctrl{1} & \qw
    \arrow[dll,dash,controls={+(1.3,0) and +(-1.3,0)}]
    \\
    {} & \targ{} & \rstick{$f_\mathrm{out}$}\qw
  \end{quantikz}
  \caption{Causality violation in naive higher-order quantum control in \eqref{eq:prog-example-qif-causal-violate}.
    (Left) The function $f$ acts as a NOT gate controlled by the qubit $x$.
    Each $f_\mathrm{in}$ and $f_\mathrm{out}$ represents the input and output of $f$.
    (Right) Applying $f$ to the control qubit itself creates a closed loop.}
  \label{fig:causality-violation}
\end{figure}
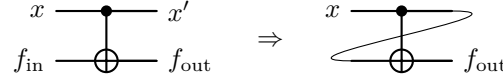

The following program reveals a critical issue:
\begin{equation*}
  \begin{aligned}
     &
    \Let
    {x' \otimes f \colon \qbit \otimes (\qbit \li \qbit)}
    = \big(
    \qifx{ x }{ \gX }{ \id }\,
    \big)
    \In { f(x') } \colon \qbit
    .
  \end{aligned}
  \tag{$**$}
  \label{eq:prog-example-qif-causal-violate}
\end{equation*}
In this program,
the $\qif$ term first generates a $\gX$ gate controlled by a qubit $x$, that is $\gCX$ gate.
It assigns the target $\gX$ gate to $f$, and renames the control qubit to $x'$.
Then, it applies $f$ to $x'$ itself.
This implies that the qubit $x$ used for control must again be used
as the target qubit of the $\gCX$ gate.
As illustrated in \cref{fig:causality-violation},
this results in a closed timelike curve, violating causality.

It is also possible to explain the problem of the program~\eqref{eq:prog-example-qif-causal-violate} without appealing to diagrammatic intuition.
Let \( \Hilb \) be the category of finite-dimensional Hilbert spaces and all linear maps.
Since \( \Hilb \) is a compact closed category and has the coproduct \( \qbit = 1 + 1 \), the above program can be interpreted in \( \Hilb \) in the standard way. The resulting denotation is a linear map \( \CC^2 \longrightarrow \CC^2 \).
We give an explicit description of this linear map.
For \( x = \ket{0} \), we have \( x' = \ket{0} \) and \( f = \id \), so the outcome is \( f(x') = \ket{0} \).
Similarly, for \( x = \ket{1} \), we have \( x' = \ket{1} \) and \( f = \gX \), so the outcome is \( f(x') = \ket{0} \).
So the denotation of the above program is
\(
    \begin{pmatrix}
    1 & 0 \\
    1 & 0 \\
    \end{pmatrix}
\).
However, since physically realizable linear maps of type \( \qbit \multimap \qbit \) are only unitaries, the above program is not physically implementable.

\subparagraph{Simmons and Kissinger's Causal Logic.}

The above discussion shows that, when the type of the control target \( A \) is a function type, allowing interaction between the control qubit and the controlled function leads to a problem.
Hence, we need a mechanism that properly regulates such interaction.

To clarify what kind of mechanism would be required, let us examine how quantum conditionals controlling functions could be implemented.
A point is that, although it is unclear how to implement \( \qifx{ x }{ \gX }{ \id } \), it is clear what is \( \qifx{ x }{ \gX\,L }{ \id\,L } \) for given \( L \), which is the controlled-\( \gX \) gate applied to \( (x,L) \).
Our idea is that a quantum conditional branching that controls functions should not be executed immediately; instead, its execution should be deferred until the argument of the controlled function is determined.
Guided by this intuition, we design a type system that avoids the above problem.

The above discussion suggests that the type of \(\qifx{M}{N_1}{N_2}\) should not be \( \qbit \otimes A \), but rather a type expressing that ``\(\qbit\) becomes available only after \( A \).''
Fortunately, a logical connective with precisely this intended meaning have already been studied, namely the \emph{seq-connective} \( \seq \) from pomset logic~\cite{retore_pomset_1997} and
BV-logic~\cite{Guglielmi2007}.
These logics are both based on
multiplicative linear logic (MLL)
with MIX rule,
conservatively extended by adding
$A \seq B$ whose intuitive meaning is
the causal relation
``$A$ happens \emph{before} $B$''.

By adopting this idea, we can now fix the typing rule~\eqref{rule:wrong-qif}
to the following by restricting the usage of the control qubit
to only after $A$ is ``resolved'' (where $A \after B \defeq B \before A$):
\begin{proofrule}
  \infer*{
    \Gamma \vdash M \colon \qbit
    \\
    \Gamma' \vdash N_1 \colon A
    \\
    \Gamma' \vdash N_2 \colon A
  }{
    \Gamma, \Gamma' \vdash \qifx{M}{N_1}{N_2} \colon \qbit \after A
  }.
\end{proofrule}
We shall prove that this new typing rule actually solves the problem.

However, this mechanism is too restrictive on its own.
For example, the following circuit
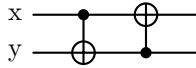
\[
  \begin{quantikz}[row sep = 0.5em]
    \lstick{x} & \ctrl{1} & \targ{ } & \\
    \lstick{y} & \targ{ } & \ctrl{-1} &
  \end{quantikz}
\]
is obviously physically realizable, but the corresponding program
\begin{equation*}
  \letx{(x', y')}{
    \big(\qifx{x}{\gX\,y}{y}\big)
  }{
    \big(\qifx{y'}{\gX\,x'}{x'}\big)
  }
\end{equation*}
is not well-typed,
since for $(x', y') \colon \qbit \after \qbit$, the first \(\qbit\) \( x' \) will be only available after the second \( \qbit \) \( y' \) has been consumed.

Here, a concept of \emph{first-order proposition} introduced by Simmons and Kissinger~\cite{Simmons2024} comes into play.
The idea is to equip a ``causally simple'' proposition with the attribute ``first-order'' and to admit the following additional rules for first-order propositions (where $F$ is first order and $A$ is arbitrary):
\begin{equation*}
  F \otimes A \iff F \before A
  \qquad
  F \parr A \iff F \after A.
\end{equation*}
Simmons and Kissinger's \emph{causal logic}~\cite{Simmons2024} is pomset logic~\cite{retore_pomset_1997} enriched with the notion of first-order propositions.
In causal logic,
by considering $\qbit$ as a first order type,
we have the type isomorphism
$\qbit \otimes \qbit \iff \qbit \after \qbit$,
which solves our concern.

Based on this intuition, this paper proposes a higher-order quantum programming language with quantum conditional branching, equipped with a type system grounded in causal logic.
For technical reasons, our type system is based on a weaker version of causal logic, namely \emph{intuitionistic BV logic}~\cite{AcclavioS25} extended with the notion of first-order propositions.
We show that the type system behaves as intended: every program of type \( \qbit \multimap \qbit \) denotes a unitary transformation (\cref{prop:semantic-soundness-of-the-type-system}).

Our type system is important not only because it resolves the aforementioned issue concerning causality, but also because it provides new evidence for the relevance of causality-aware logics to higher-order quantum computation.
Previous work showed that structures of causality-aware logics can be found in causal models motivated by quantum computation.
This is undoubtedly an important observation.
However, it was not clear whether these structures appeared merely incidentally, or whether they are essential to higher-order quantum computation.
This paper demonstrates that causality-aware logics can play an active role in solving a concrete problem, thereby providing new evidence for their relevance.

\subparagraph{Categorical Model.}

We also construct a new categorical model.
This model is used to prove the soundness of our type system, and it also answers the previously-mentioned question of whether one can build a model consisting solely of pure supermaps.

Our construction is motivated by the \emph{Caus construction}~\cite{Kissinger2019,Simmons2022,Simmons2024}.
It is a categorical construction
that produces a model of causal logic
(hence models BV-logic and pomset-logic)
from a compact closed category with some additional
structure, called discarding maps
$\discardDiag_A\colon A \longrightarrow I$.
A particular example of such category is $\CPM$,
which is the category of completely positive maps.
By the Caus construction, we can construct the category
$\CausCPM$ that can capture the causal structure of supermaps.

However, this category is insufficient
for two reasons:
(i) it does not rule out non-pure supermaps, such as the OCB process; and
(ii) it does not interpret quantum conditionals.
In particular, the latter is critical.
Let us consider the program
$
  \qifx{x}{U\, y}{y}
$,
whose semantics should be the controlled unitary $CU$.
To define a compositional semantics,
we should be able to construct $CU$ from $U$,
but the map
$U \mapsto CU = (\ketbra{1}{1} \otimes U) + (\ketbra{0}{0}\otimes \id)$
turns out be not completely positive
because it respects global phases.

Our idea is to leverage the connection between \( \Hilb \) and \( \CPM \): we define \(\CausHilb\) as the pullback of \(\Caus[\CPM] \longrightarrow \CPM \) via \( \Hilb \longrightarrow \CPM \).
A high-level categorical argument establishes that \(\CausHilb\) carries the expected logical structure.

\subparagraph{Related Work.}

Some closely related work has already been discussed above.

The quantum SWITCH was first discovered in \cite{Chiribella2013},
and later, many different kinds of supermaps with
indefinite causal order have been found,
including
the OCB process~\cite{Oreshkov2012-ocb},
the Lugano process~\cite{BaumelerFF14,Baumeler_2016},
and the Grenoble process~\cite{Wechs2021}.

The class of supermaps that are smaller than pure supermaps
includes
extensibly causal maps~\cite{OreshkovG2016,FeixAB2016}
which are maps that do not violate
causal inequality~\cite{AbbottGCB16,Branciard_2016,Baumeler_2016},
and 
QC-QC~\cite{Wechs2021}
which is defined as a specific procedure
to create a superposition of causal order.
The QC-QC is included in the class of extensibly causal maps,
but equality of these two classes remains an
open problem.
Our approach is also different from these studies
in that we can describe classes of higher-order functions.
We also conjecture that the supermaps expressible in the language we present
define a subclass of these two classes.

Pomset-logic was introduced in 1997~\cite{retore_pomset_1997,retore_pomset_2020}
as a logic whose correctness is defined by proof-nets.
On the other hand, BV-logic was introduced in
the series of papers
``A system of interaction and structure''~\cite{Guglielmi2007,Tiu06,NguyenS23,BurgerG11,GuglielmiS11}.
The first paper~\cite{Guglielmi2007} introduced BV-logic in 2007,
and the second paper~\cite{Tiu06} discussed the necessity of deep inference.
The third paper was initially intended to prove the equivalence of pomset-logic and BV-logic.
However, although this had long been believed, it was disproved considerably later,
and this negative result for the conjecture became
the third paper~\cite{NguyenS23} in 2023.
The fourth~\cite{BurgerG11} and fifth~\cite{GuglielmiS11} papers concern exponentials.
An intuitionistic variant of BV-logic was published recently in~\cite{AcclavioS25}.
There is prior work relating these logics to the $\pi$-calculus
\cite{acclavio2025proofsexecutiontreespicalculus,HorneT19},
but, to the best of our knowledge,
this is the first paper to incorporate BV-logic into
a programming language via the Curry--Howard correspondence
in a setting not based on the $\pi$-calculus.

Apart from the Caus construction,
there are some categorical studies of supermaps.
In~\cite{Hefford2024},
Heffords and Wilson have studied semantics
of supermaps based on strong profunctors,
and in \cite{BV-model-HeffordWilson_2025},
they also studied the Chu construction~\cite{Barr91},
which is known as a general construction that creates
a $\ast$-autonomous category from a monoidal closed category,
and yields a BV-category from a duoidal category.
Other than that, Li and Zamdzhiev~\cite{QCS-revisited} have studied
the BV-structure in the category of operator spaces.

There is also a study that relates causal relations in physics with linear logic
in \cite{projective-causality-hoffreumon_2024,jencova2026}.
These work can be compared with causal logic~\cite{Simmons2024}.
These studies appear to be based on similar concepts,
yet the structures of union and intersection seem to differ.

\subparagraph{Outline.}

We first review some preliminaries in \cref{sec:preliminary}, and review BV logic in
\cref{sec:bv}.
In \cref{sec:lang}, we introduce our language \ourlang{} together with a simple
(degenerate) denotational semantics in \(\Hilb\).
In \cref{sec:CausHilb}, we define our categorical model of pure supermaps and show that it is a BV-category.
Finally, in \cref{sec:cat-sem-causHilb}, we study the semantics of \ourlang{} in detail:
we prove a full-abstraction theorem and investigate term definability.

\section{Preliminaries}
\label{sec:preliminary}

\subsection{Linear Logic: Proof Systems and Models}

In this section, we briefly review the terminology and background
on linear logic that will be used throughout the paper.
For further details, we refer the reader to the standard literature.

Linear logic, introduced by Girard~\cite{GIRARD19871},
is a logic in which the structural rules of weakening and contraction
are not freely available.
This restriction corresponds to disallowing discarding and copying,
and for this reason linear logic is often described as
a \emph{resource-aware} logic.

Linear logic has several well-known fragments.
One of them is multiplicative intuitionistic linear logic (MILL),
which features the connectives
$\otimes$, $\li$ and the unit $I$.
One may view $\otimes$ and $\li$ as linear analogues of pairing and functions,
and compare them with conjunction $\land$ and implication $\to$
in ordinary logic.
A crucial difference, however, is that
$\otimes$ is not idempotent, \ie $A$ and $A \otimes A$
represent different amounts of resources.
In particular,
$A \otimes B \vdash A$ is not derivable, since $B$ cannot be discarded.
The Curry--Howard--Lambek correspondence between MILL,
the linear lambda calculus, and symmetric monoidal closed categories is well known.

A larger fragment containing MILL is multiplicative linear logic (MLL).
It can be seen as extending MILL with negation
$A^\bot$, which allows one to define
$\llpar$, the De Morgan dual of $\otimes$, as
$A \llpar B \defeq (A^\bot \otimes B^\bot)^\bot$.
The categorical structure corresponding to MLL is a
\emph{$\ast$-autonomous category}, which is defined as follows.
\begin{definition}
  A $\ast$-autonomous category is
  a symmetric monoidal closed category $(\categoryC, \otimes, I, \li)$
  equipped with a fully faithful functor
  $(-)^* \colon \categoryC^\op \longrightarrow \categoryC$
  and a natural isomorphism
  $A \li B^* \iso (A \otimes B)^*$.
  \lipicsEnd
\end{definition}

In particular, among $\ast$-autonomous categories,
there is a somewhat degenerate class
known as \emph{compact closed categories}.
These are categories in which the two
monoidal structures $\otimes$ and $\parr$ coincide.
They can be defined as a monoidal category with
\emph{unit} $I \longrightarrow A^* \otimes A$
and \emph{counit} $A \otimes A^* \longrightarrow I$
satisfying some axioms.

It is also common to extend MLL with the rule
$
  A \otimes B \vdash A \llpar B
$.
This rule is called the \emph{MIX} rule.
Categorically, it corresponds to equipping a
\(*\)-autonomous category with a morphism \(\bot \to I\) satisfying the appropriate
coherence conditions.
In particular, when this morphism is an isomorphism, such a
category is called an \emph{isoMIX category}~\cite{cockett1997proof}.

\subsection{Basics of Quantum Computation and Supermaps}

In this section, we briefly review basic notions from quantum computation.
For precise definitions, see, for example, \cite{Selinger2004a,Selinger2004}.

A (pure) quantum state is represented by a norm-\(1\) element of a Hilbert space.
In this paper we restrict attention to the finite-dimensional setting,
so a pure state is a vector \(v \in \CC^n\) with \(\|v\|=1\).
We write the standard basis vectors \(e_0,\dots,e_{n-1} \in \CC^n\) as
\(\ket{0},\dots,\ket{n-1}\).
For qubits, the basis states \(\ket{0}\) and \(\ket{1}\) are often identified with
the Booleans \texttt{false} and \texttt{true}.
Unlike the classical case, a qubit state need not be \(\ket{0}\) or \(\ket{1}\);
it can also be any linear combination \(\alpha\ket{0}+\beta\ket{1}\),
called a \emph{superposition} state.
In quantum computation one also considers \emph{mixed states}, which represent
probabilistic mixtures of pure states.
Mathematically, a mixed state on \(\CC^n\) is a \emph{density matrix},
\ie a positive semidefinite operator \(\rho \in \CC^{n\times n}\) with
\(\Tr(\rho)=1\).
Equivalently, \(\rho\) can be written as a convex combination of rank-one projectors:
$
  \rho \;=\; \sum_i p_i \ketbra{\psi_i}{\psi_i}
$
where
$
  p_i \ge 0,\ \sum_i p_i = 1,\ \|\psi_i\|=1.
$
A pure state \(\ket{\psi}\) corresponds to the special case
\(\rho = \ketbra{\psi}{\psi}\), which has rank \(1\).

To describe general state transformations, we use \emph{quantum channels}.
Formally, a quantum channel is a completely positive trace-preserving (CPTP)
linear map
\(\Phi : \mathcal{L}(\CC^n) \to \mathcal{L}(\CC^m)\),
where \(\mathcal{L}(\CC^n)\) denotes the set of $n \times n$ matrices.
A particularly important subclass consists of channels that preserve purity.
Such a channel is implemented by an \emph{isometry} \(V : \CC^n \to \CC^m\)
satisfying \(V^\dagger V = I\), via \(\Phi(\rho) = V \rho V^\dagger\).
We refer to these as \emph{pure channels}.
Every channels can be decomposed to pure channels followed
by some qubit discardings.

We also consider \emph{higher-order} transformations that take quantum channels as inputs.
In particular, it should preserve the trace-preserving property.
Moreover, when represented in a matrix form
(e.g.\ via a Choi--Jamio\l{}kowski-type representation),
the supermap itself is \emph{completely positive}.
Following the definition in \cite{Yokojima2021consequencesof},
we define \emph{pure supermaps} as a linear map
$
  S\colon
  \big(\bigotimes_{n = 1}^N \Lin(A_i) \li \Lin(B_i)\big)
  \longrightarrow \Lin(A) \li \Lin(B)
$
such that for all auxiliary (finite-dimensional) Hilbert spaces
$X_i$ and $Y_i$
and isometry operators
$f_i \colon A_i \otimes X_i \li B_i \otimes Y_i$,
the partial application to $S$
described as follows
defines an isometry
$A \otimes_i X_i \li B \otimes_i Y_i$:
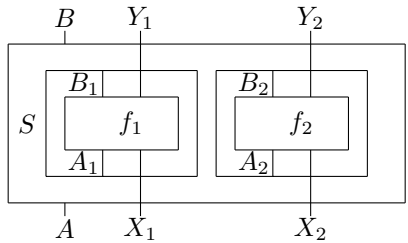
\[
\begin{tikzpicture}[yscale=0.7]
	\begin{pgfonlayer}{nodelayer}
		\node [style=none] (0) at (5, -1.5) {};
		\node [style=none] (1) at (-0.25, -1.5) {};
		\node [style=none] (3) at (0.25, -1) {};
		\node [style=none] (4) at (0.25, 1) {};
		\node [style=none] (6) at (5, 1.5) {};
		\node [style=none] (7) at (-0.25, 1.5) {};
		\node [style=none] (8) at (1, -1) {};
		\node [style=none] (9) at (1, 1) {};
		\node [style=none] (10) at (1, -0.5) {};
		\node [style=none] (11) at (1, 0.5) {};
		\node [style=none] (12) at (0.75, -0.75) {$A_1$};
		\node [style=none] (13) at (0.75, 0.75) {$B_1$};
		\node [style=none] (14) at (0.5, 1.5) {};
		\node [style=none] (15) at (0.5, 1.75) {};
		\node [style=none] (16) at (0.5, 2) {$B$};
		\node [style=none] (17) at (2, -0.5) {};
		\node [style=none] (18) at (0.5, -0.5) {};
		\node [style=none] (19) at (0.5, 0.5) {};
		\node [style=none] (20) at (2, 0.5) {};
		\node [style=none] (21) at (1.5, -1.75) {};
		\node [style=none] (22) at (1.5, -0.5) {};
		\node [style=none] (23) at (1.5, 0.5) {};
		\node [style=none] (24) at (1.5, 1.75) {};
		\node [style=none] (25) at (1.5, -2) {$X_1$};
		\node [style=none] (26) at (1.5, 2) {$Y_1$};
		\node [style=none] (27) at (0, 0) {$S$};
		\node [style=none] (28) at (1.375, 0) {$f_1$};
		\node [style=none] (29) at (0.5, -1.75) {};
		\node [style=none] (30) at (0.5, -1.5) {};
		\node [style=none] (31) at (0.5, -2) {$A$};
		\node [style=none] (33) at (3.25, -1) {};
		\node [style=none] (34) at (3.25, -0.5) {};
		\node [style=none] (35) at (3.25, 0.5) {};
		\node [style=none] (36) at (3, -0.75) {$A_2$};
		\node [style=none] (37) at (4.25, -0.5) {};
		\node [style=none] (38) at (2.75, -0.5) {};
		\node [style=none] (39) at (2.75, 0.5) {};
		\node [style=none] (40) at (4.25, 0.5) {};
		\node [style=none] (41) at (3.75, -0.5) {};
		\node [style=none] (43) at (3.625, 0) {$f_2$};
		\node [style=none] (44) at (3, 0.75) {$B_2$};
		\node [style=none] (45) at (3.75, -1.75) {};
		\node [style=none] (46) at (3.75, -0.5) {};
		\node [style=none] (47) at (3.75, -2) {$X_2$};
		\node [style=none] (48) at (3.75, 0.5) {};
		\node [style=none] (49) at (3.75, 1.75) {};
		\node [style=none] (50) at (3.75, 2) {$Y_2$};
		\node [style=none] (51) at (2.25, 1) {};
		\node [style=none] (53) at (3.25, 1) {};
		\node [style=none] (54) at (2.25, -1) {};
		\node [style=none] (55) at (2.5, 1) {};
		\node [style=none] (56) at (2.5, -1) {};
		\node [style=none] (57) at (4.5, -1) {};
		\node [style=none] (58) at (4.5, 1) {};
	\end{pgfonlayer}
	\begin{pgfonlayer}{edgelayer}
		\draw (0.center) to (1.center);
		\draw (1.center) to (7.center);
		\draw (7.center) to (6.center);
		\draw (4.center) to (3.center);
		\draw (3.center) to (3.center);
		\draw (8.center) to (8.center);
		\draw (8.center) to (10.center);
		\draw (11.center) to (9.center);
		\draw (14.center) to (15.center);
		\draw (17.center) to (20.center);
		\draw (20.center) to (19.center);
		\draw (19.center) to (18.center);
		\draw (18.center) to (17.center);
		\draw (23.center) to (24.center);
		\draw (21.center) to (22.center);
		\draw (29.center) to (30.center);
		\draw (33.center) to (33.center);
		\draw (33.center) to (34.center);
		\draw (37.center) to (40.center);
		\draw (40.center) to (39.center);
		\draw (39.center) to (38.center);
		\draw (38.center) to (37.center);
		\draw (45.center) to (46.center);
		\draw (48.center) to (49.center);
		\draw (51.center) to (4.center);
		\draw (0.center) to (6.center);
		\draw (35.center) to (53.center);
		\draw (54.center) to (3.center);
		\draw (54.center) to (51.center);
		\draw (56.center) to (55.center);
		\draw (57.center) to (56.center);
		\draw (58.center) to (55.center);
		\draw (57.center) to (58.center);
		\draw (27.center) to (27.center);
	\end{pgfonlayer}
\end{tikzpicture}
.
\]

\subsection{Definition of Basic Categories}

We define two basic categories $\Hilb$ and $\CPM$
which are both standard model of quantum computation.
Both categories are compact closed.

\begin{definition}
  The category $\Hilb$ consists of the set of objects
  $\Ob(\Hilb) = \Nat$ where the maps
  $\Hilb(n, m)$ is defined by the linear maps
  $\CC^n \longrightarrow \CC^m$.

  This category $\Hilb$ is compact closed where the structures
  are defined as follows:
  \begin{itemize}
    \item The monoidal unit is $1$, and the monoidal product of objects $n \otimes m$ is $n \times m$.
    \item The monoidal product of maps are defined by the tensor product,
          \ie for linear maps \( f \in \Hilb(n,n') \) and \( g \in \Hilb(m,m') \), when \( f \ket{i} = \sum_{i'} \alpha_{i'} \ket{i'} \) and \( g \ket{j} = \sum_{j'} \beta_{j'} \ket{j'} \), let \( (f \otimes g)(\ket{i} \otimes \ket{j}) := \sum_{i',j'} \alpha_{i'} \beta_{j'} (\ket{i'} \otimes \ket{j'}) \).
    \item The monoidal structure is strict: the associater
          $\alpha \colon (A \otimes B) \otimes C
            \longrightarrow A \otimes (B \otimes C)$
          and the uniters
          $\ell \colon I \otimes A \longrightarrow A$ and
          $r \colon A \otimes I \longrightarrow A$ are identity.
    \item The dual of an object is itself $n^* = n$,
          and the dual map $f^* \colon m \longrightarrow n$
          of a map $f \colon n \longrightarrow m$ is defined by the transpose.
    \item The unit map
          $\eta \colon 1 \longrightarrow n \otimes n$
          and the counit map
          $\varepsilon \colon n \otimes n \longrightarrow 1$
          are defined as follows:
          \begin{align*}
            \eta \colon
            1 \longmapsto \textstyle\sum_{i = 0}^n \ket i \otimes \ket i,
            \qquad
            \varepsilon \colon
            \ket i \otimes \ket j \longmapsto \delta_{ij}.
          \end{align*}
  \end{itemize}
  It also admits finite biproducts defined by
  $n \oplus_\Hilb m \defeq n + m$.
  \lipicsEnd
\end{definition}

\begin{definition}
  The category $\CPM$ has sequences of natural numbers
  $\vec n = (n_1, \dots, n_k)$ as objects,
  and each map $f \in \CPM(\vec n, \vec m)$
  is defined by a matrix of maps
  $
    ( f_{ij} ) \colon \vec n \longrightarrow \vec m
  $
  where
  $f_{ij} \colon \Lin(\CC^{n_i}) \longrightarrow \Lin(\CC^{m_j})$
  is completely positive.
  The composition is defined by
  the matrix multiplication, \ie
  $(g_{k\ell}) \circ (f_{ij}) = (\sum_i g_{j\ell} \circ f_{ij})$.

  This category is also compact closed with the following structures:
  \begin{itemize}
    \item The monoidal unit is $(1)$, and the monoidal product $\vec n \otimes \vec m$ is given by
          \[
            (n_1m_1,\, n_1m_2,\, \dots,\, n_1m_\ell,\, n_2m_1,\, \dots,\, n_km_\ell).
          \]
    \item The monoidal product of morphisms $(f_{ij}) \otimes (g_{i'j'})$
          is defined by the pointwise Kronecker product of CPMs
          $(f_{ij} \otimes g_{i'j'})$.
    \item It is strict monoidal.
    \item The dual of an object is itself ${\vec n}^* = \vec n$,
          and the dual of $(f_{ij})$ is defined by
          $(f^*_{ji})$ where $f^*_{ji}$ is
          the transpose of $f_{ij}$.
    \item The unit and counit are defined by
          $(\eta_{1,ij})$ and $(\varepsilon_{ij,1})$ where
          \begin{align*}
            \eta_{1,ij}        &
            = \begin{cases}
                1 \longmapsto \sum_{k,\ell \in \{0,\dots,n_i - 1\}}
                \ketbra{k}{\ell} \otimes \ketbra{k}{\ell}
                                & \text{if } i = j
                \\
                1 \longmapsto 0 & \text{otherwise},
              \end{cases}
            \\
            \varepsilon_{ij,1} &
            = \begin{cases}
                \ketbra{k}{\ell} \otimes \ketbra{k'}{\ell'}
                \longmapsto
                \delta_{k,k'} \delta_{\ell\ell'}
                  & \text{if } i = j   \\
                \ketbra{k}{\ell} \otimes \ketbra{k'}{\ell'}
                \longmapsto
                0 & \text{otherwise}.
              \end{cases}
          \end{align*}
  \end{itemize}
  It also admits finite biproducts defined by
  the concatination of sequences $\vec{n}\oplus_\CPM\vec{m} = \vec{n}\vec{m}$.
  \lipicsEnd
\end{definition}

These two categories are related with the following
\emph{embedding functor}.\footnote{
  We call this functor embedding
  since it embeds pure quantum computation
  into general quantum computation
  even though it is not full or faithful.
}

\begin{definition}
  There is a functor $\EmbeddingFunctor \colon \Hilb \longrightarrow \CPM$
  that sends $n$ to $(n)$
  and $f$ to the map
  $\lambda \rho.\ f \circ \rho \circ f^\dagger$.
  This functor preserves all the compact closed strucuture
  \emph{on the nose},
  but it does not preserve biproducts.
  \lipicsEnd
\end{definition}

Since these two categories are compact closed,
they can model a linear lambda calculus.
The category $\Hilb$ can interpret pure quantum computation:
the semantics of the qubit type can be defined by $2 = 1 \oplus_\Hilb 1$,
where a state of a qubit $I = 1 \longrightarrow 2$ corresponds with a vector,
and the semantics of a program can be given as a linear functions.
On the other hand, $\CPM$ can interpret any quantum computation
including measuremets or discardings.
Through the embedding functor $\EmbeddingFunctor$,
the pure quantum computation can also be interpreted in $\CPM$.
Additionally, the semantics of Booleans can be defined by
$(1, 1) = 1 \oplus_\CPM 1$,
and the measurement and discarding maps are defined as follows:
\begin{align*}
  \begin{array}{rccc}
    \measmap \defeq (P_i)_{i = 1,2} \colon
     & (2)
     & \longrightarrow
     & (1, 1)
    \\
    \text{where}\ \ P_i \colon
     & M
     & \longmapsto
     & M_{ii},
  \end{array}
  \qquad
  \begin{array}{rccc}
    \discardDiag_n \defeq (\Tr_{n_i})_{i = 1,\dots,k} \colon
     & \vec n
     & \longrightarrow
     & (1)
    \\
    \text{where}\ \ \Tr_{m} \colon
     & M
     & \longmapsto
     & \sum_{j = 1}^{m} M_{jj}.
  \end{array}
\end{align*}

\section{BV-logic}
\label{sec:bv}

In addition to linear logic,
BV-logic~\cite{Guglielmi2007} and pomset-logic~\cite{retore_pomset_1997}
introduce a \emph{non-commutative}
and \emph{self-dual} connective $\seq$ called \emph{seq} or \emph{before}
to the logic.
That is, for each pair of propositions $A$ and $B$,
unlike other connectives such as $\otimes$, $\llpar$, $\&$ or $\oplus$,
there is no symmetry rule $A \seq B \implies B \seq A$,
and it admits De Morgan duality with itself
$(A \seq B)^\bot \iff A^\bot \seq B^\bot$.
Both BV-logic and pomset-logic are conservative extensions
of MLL + nullary and binary MIX.

BV-logic can be seen as a linear logic with causality relation,
where $A \seq B$ means $A$ happens before $B$.
We list below two illustrative derivations in BV.
\begin{gather*}
  A \otimes B\ \vdash\ A \seq B\ \vdash\ A \llpar B
  \qquad\quad
  (A \seq C) \otimes (B \seq D)\ \vdash\ (A \otimes B) \seq (C \otimes D)
\end{gather*}
In linear logic, $A \otimes B$ can be understood as a pair
of resources $A$ and $B$.
Here, in BV-logic, we can additionally assume
$A \otimes B$ as a pair \emph{without causal ordering}.
Therefore, the first derivation above $A \otimes B \vdash A \seq B$
can be interpreted as
``from a pair of resources $A \otimes B$,
we can force a causal relation to obtain $A \seq B$''.
Dually, $A \llpar B$ can be understood as a pair
\emph{with some causation}.
So $A \seq B \vdash A \llpar B$ can be interpreted as
``by forgetting the detail of causal relation $A \seq B$,
we obtain a pair of resource with unknown causation $A \llpar B$''.
The second is the \emph{interchange rule},
which can also be understood similarly as,
``by assuming extra causal relations from $A$ to $D$ and from $B$ to $C$,
we can obtain $(A \otimes B) \seq (C \otimes D)$ from $A \seq C$ and $B \seq D$''.
Similar interchange rules also appear in other logic literature
that addresses causality in a context of concurrency~\cite{HoareMSW11,PaquetS24}.

BV-logic is famous for its style of inference rules
called \emph{deep inference}~\cite{Tiu06}.
Deep inference can be seen as a special kind of sequent calculus,
and adopted by BV-logic to admit cut elimination.
While usual sequent calculi have rewriting rules that
only pattern match the outermost connective of a formula,
deep inference allows one to pattern match a formula recursively at arbitrary depth,
described as follows:
\begin{proofrule}
  \infer{
    A \vdash B
  }{
    S\{ A \} \vdash S\{ B \}
  }
  \qquad \text{where} \qquad
  \begin{aligned}
    S\{\cdot\} &
    \Coloneqq \{\cdot\}
    \sor A \otimes S
    \sor S \otimes A
    \sor A \seq S
    \\ & \ \ \sor
    S \seq A
    \sor A \llpar S
    \sor S \llpar A
  .
  \end{aligned}
\end{proofrule}

\subsection{Intuitionistic BV}

Since we use a logic in this paper to define a type system for a
lambda calculus, it is natural, via the Curry-Howard correspondence,
to consider the intuitionistic variant of BV logic.

Intuitionistic BV-logic (IBV) was introduced in \cite{AcclavioS25}.
In IBV, we have $\li$ instead of $\llpar$ as in MILL.
IBV is proven to be a conservative extension of MILL.
It is also defined by deep inference rules so that 
it admits cut elimination.

We present a simplified sequent calculus \emph{sequentIBV}
for IBV in \cref{fig:IBV-sequent}.
Instead of having deep inference rules to admit cut-elimination,
we rather define it as a usual sequent calculus.
This also allows us to significantly reduce the number of rules
which were needed in IBV to comprehensively represent
all normal forms.
As we will see in the next section,
the correspondence between proofs and programs 
becomes much clearer in this simpler sequent calculus style definition.
\begin{figure}
  \begin{proofrules}[.5em]
    \infer
    {}
    {A \vdash A}

    \infer
    {\Gamma, A, B \vdash C}
    {\Gamma, A \otimes B \vdash C}

    \infer
    {\Gamma \vdash A\\
    \Gamma' \vdash B}
    {\Gamma, \Gamma' \vdash A \otimes B}

    \infer
    {\Gamma, A \vdash B}
    {\Gamma \vdash A \li B}

    \infer
    {\Gamma \vdash A\\
    \Gamma', B \vdash C}
    {\Gamma, \Gamma', A \li B \vdash C}

    \infer{\Gamma \vdash A}
    {\Gamma, I \vdash A}

    \infer{}
    {\ \ \vdash I}

    \infer
    {\Gamma \vdash A\\
    \Gamma', A \vdash B}
    {\Gamma, \Gamma' \vdash B}

    \infer
    {}
    {\ \  \vdash I \seq I}

    \infer
    {\Gamma, A \vdash B\\
    \Gamma', C \vdash D}
    {\Gamma, \Gamma', A \seq C \vdash B \seq D}

    \infer
    {
      \Gamma \vdash A \seq C\\
      \Gamma' \vdash B \seq D
    }{
      \Gamma, \Gamma' \vdash
      (A \otimes B) \seq (C \otimes D)
    }

    \infer
    {
      \Gamma \vdash (A \li C) \seq (B \li D)
    }{
      \Gamma
      \vdash
      (A \seq B) \li (C \seq D)
    }

    \infer{}{
      (A \seq B) \seq C
      \dashvdash
      A \seq (B \seq C)
    }
  \end{proofrules}
  \caption{Derivation rules in sequentIBV.}
  \label{fig:IBV-sequent}
\end{figure}

As expected, sequentIBV and IBV have the same expressivity.
See \cref{app:sec:logic-lang}.

\begin{proposition}\label{prop:IBV-eq-sequentIBV}
  $A \vdash B$ is derivable in sequentIBV
  if and only if $A \li B$ is derivable in IBV~\cite{AcclavioS25}.
\end{proposition}
\begin{proof}[Sketch of Proof]
  Checking that all sequents in IBV are provable in sequentIBV is easy.
  One can also check that all basic derivations
  in IBV can be also proven in sequentIBV.
  To check that all derivations in IBV
  that use the deep inference principle in the proof are also
  derivable in sequentIBV,
  it suffices to prove that,
  in sequentIBV,
  if $X \vdash Y$ is derivable,
  then $S\{X\} \vdash S\{Y\}$
  (or $S\{Y\} \vdash S\{X\}$ if the hole in $S$ is at the negative position)
  is also derivable for each following $S$ defined by
  \begin{align*}
    S\{\cdot\} &
    \Coloneqq \{\cdot\}
    \sor A \otimes S
    \sor S \otimes A
    \sor A \seq S
    \sor S \seq A
    \sor A \li S
    \sor S \li A
    .
  \end{align*}
  For example, if $S = \{\cdot\} \seq A$,
  this can be done as follows:
  \[
    \infer*{
      \infer*{\pi}{
        X \vdash Y
      }
      \\
      A \vdash A
    }{
      X \seq A \vdash Y \seq A
    }
    .
    \tag*{\qedhere}
  \]
\end{proof}

\subsection{First Order Proposition in Causal Logic}

Simmons and Kissinger~\cite{Simmons2024} defined \emph{causal logic}
as an extension of pomset-logic obtained
by adding \emph{first-order propositions}.
This is a particularly natural extension:
for example, in the proof-net style proof of linear logic,
it is shown that, in an appropriate sense, adding the seq-connective
is equivalent to adding first-order propositions
\cite[Proposition 4.13]{Simmons2024}.
Concretely, these first-order propositions satisfy the following.
\begin{align*}
  \forall A\colon \text{first-order},
  \quad
  A \otimes B \dashv\vdash A \seq B
  \quad \text{and} \quad
  B \seq A \dashv\vdash B \llpar A
\end{align*}
Recalling that $A \otimes B \vdash A \seq B \vdash A \llpar B$
is already derivable in BV or pomset-logic,
this means we can prove the converse if we assume first-orderness
for some propositions.
Such first-order propositions can be seen as
``data without input'' or ``data without past''.
So, when $A$ is first-order,
the causal relation $A \seq B$ says
``$B$ can be obtained after $A$ is obtained'',
but since there is nothing to wait to obtain $A$,
this is simply the same as having a pair $A\otimes B$.

In particular, in our intuitionistic setting,
we can replace
$B \seq A \dashv\vdash B \llpar A$
with
$B \li (C \seq A) \dashv\vdash (B \li C) \seq A$,
where the left-to-right implication is always possible.

\section{Language for Pure Quantum Computation}
\label{sec:lang}

We now present our language \ourlang{}.
Our language is \emph{pure} (measurement-free),
\emph{higher-order}, and has \emph{quantum conditional branching}.
We also have a type system inspired by IBV,
so we have the $\seq$ type representing causal ordering.

\subsection{Syntax}

We define the syntax of our language \ourlang{} in \cref{fig:lang-syntax}.
\begin{figure}
  \begin{gather*}
  \begin{aligned}
      \textit{Types} \quad A,B \quad \Coloneqq \quad&
       d\ (\in \Nat^+) \sor \bot \sor A \otimes B \sor A \seq B \sor A \li B
      \\[5pt]
      \textit{First-order types}\quad F : \FO \quad \Coloneqq \quad&
      d \sor \bot \sor F \otimes F \sor F \seq F
      \\[5pt]
      \textit{Terms} \quad M,N \quad \Coloneqq \quad&
       x \sor \unitval \sor \lambda x.\, M \sor M\,N \sor M \otimes N
      \\ \qquad\sor\quad
       & \letx{\ \_}{M}{N} \sor \letx{x \otimes y}{M}{N}
      \\ \qquad\sor\quad
       & M \seq N \sor \letx{x \seq y}{M}{(N_1 \seq N_2)}
      \\ \qquad\sor\quad
       & \sassocL(M) \sor \sassocR(M)
        \sor
       \sinterch(M) %
      \\ \qquad\sor\quad
       &
       \pawait(M) \sor \pexpose(M)
      \\ \qquad\sor\quad
       & \ceil{f} \sor \qinj \sor \gU \sor \qifx{M}{N_1}{N_2}
  \end{aligned}
      \\[7pt]
    \text{where $f$ in $\ceil{f}$ is a type isomorphism, defined at \eqref{eq:type-iso}.}
    \\
    \text{Abbreviations:} \quad
      \qbit \defeq 2,\quad
      (\letx{x}{M}{N}) \defeq (\lambda x.\,N) M,\quad
      A \after B \defeq B \before A.
  \end{gather*}
  \vspace{-5ex}
  \caption{The syntax of \ourlang.}
  \label{fig:lang-syntax}
\end{figure}

\subparagraph{Types.}
For the types, we have $d$ (a positive natural number)
that represents qu$d$its, which are quantum data types of $d$ dimensions.
Our language is based on Intuitionistic BV-logic,
and has
$\otimes$ (tuple type),
$\bot$ (unit type),
$\seq$ (causal ordering),
and $\li$ (function type).
We have first-order types, corresponding to first-order propositions.
The atomic types $d$ and $\bot$ are first-order,
and they are closed under $\otimes$ and $\seq$,
but not under $\li$.

\subparagraph{Terms.}
In our language \ourlang{},
we have terms that are derived from
linear lambda calculus, BV-logic, and pure quantum computation.

In the first two lines of~\cref{fig:lang-syntax},
we have some terms taken from standard linear lambda calculus.
We have
$x$ (variable),
$\unitval$ (the value of the type $\bot$),
$\lambda x.M$ (lambda abstraction),
$M\,N$ (function application),
$M \otimes N$ (pair of terms),
and two let bindings
corresponding to $\bot$ and $\otimes$.
The usual let binding can be derived as $(\lambda x.N)M$.

In the next two lines, we have some terms
derived from BV-logic.
We have $M \seq N$ (causally ordered terms)
and a let binding for $\seq$.
We also have explicit associators $\sassocR$ and $\sassocL$ for $\seq$,
which are the inverse of each other.
The term $\sinterch$ corresponds to the interchange rule in BV-logic.

The terms $\pawait$ and $\pexpose$ are the special rules for
first-order types, which will be explained along with the typing rules
in~\cref{subsec:lang-typingRules}.

In the last line, we have some quantum primitives.
We have $\qinj$ (embedding to larger space $n\to n+m$), $\gU$ (unitary map),
and $\qif$ (quantum conditional).
We also have a casting $\ceil{f}$ along a type isomorphism $f \colon A \iso B$,
where type isomorphism is defined as a congruence relation as
\begin{equation}
  1 \iso \bot,\quad
  n \otimes m \iso n \times m
  .
  \label{eq:type-iso}
\end{equation}

\subsection{Typing Rules}
\label{subsec:lang-typingRules}

The typing rules for \ourlang{} are given in~\cref{fig:typing-rule}.
\begin{figure}[t]
  \begin{proofrules}[0.3em]
    \infer
    {}{x \colon A \vdash x \colon A}

    \infer
    {}{\ \vdash \unitval \colon \bot}

    \infer
    {\Gamma, x \colon A \vdash M \colon B}
    {\Gamma \vdash \lambda x. M \colon A \li B}

    \infer
    {\Gamma \vdash M \colon A \li B
      \\
      \Gamma' \vdash N \colon A}
    {\Gamma, \Gamma' \vdash M\,N \colon B}

    \infer
    {\Gamma \vdash M \colon A\\
      \Gamma' \vdash N \colon B}
    {\Gamma, \Gamma' \vdash M \otimes N \colon A \otimes B}

    \infer
    {\Gamma \vdash M \colon \bot\\
      \Gamma' \vdash N \colon A}
    {\Gamma, \Gamma' \vdash \letx{\,\_}{M}{N} \colon A}

    \infer
    {\Gamma \vdash M \colon A \otimes B\\
      x \colon A, y \colon B, \Gamma' \vdash N \colon C}
    {\Gamma, \Gamma' \vdash \letx{x \otimes y}{M}{N} \colon C}

    \infer
    {\Gamma \vdash M \colon (A \seq B) \seq C}
    {\Gamma \vdash \sassocL(M) \colon A \seq (B \seq C)}

    \infer
    {\Gamma \vdash M \colon A \seq (B \seq C)}
    {\Gamma \vdash \sassocR(M) \colon (A \seq B) \seq C}

    \infer*[lab = seq-intro]
    {\Gamma \vdash M \colon A\\
    \Gamma' \vdash N \colon B}
    {\Gamma, \Gamma' \vdash M \seq N \colon A \seq B}
    \label{lrule:seq-intro}

    \infer*[lab = let-seq]
    {\Gamma_0 \vdash M \colon A \seq B\\
      x \colon A, \Gamma_1 \vdash N_1 \colon C\\
      y \colon B, \Gamma_2 \vdash N_2 \colon D}
    {\Gamma_0, \Gamma_1, \Gamma_2 \vdash
      \letx{x \seq y}{M}{(N_1 \seq N_2)}
      \colon C \seq D}
    \label{lrule:let-seq}

    \infer*[lab = interch]
    {\Gamma \vdash
      M \colon (A \seq C) \otimes (B \seq D)}
    {\Gamma \vdash
      \sinterch(M) \colon (A \otimes B) \seq (C \otimes D)}
    \label{lrule:interch}

    \infer*[lab = await]
    {\Gamma \vdash M \colon A \seq B\\
      A \colon \FO}
    {\Gamma \vdash \pawait(M) \colon A \otimes B}
    \label{lrule:await}

    \infer*[lab = expose]
    {\Gamma \vdash M \colon B \li (C \seq A)\\
      A \colon \FO}
    {\Gamma \vdash \pexpose(M) \colon (B \li C) \seq A}
    \label{lrule:expose}

    \infer*[lab=type-iso]
    {f \colon A \iso B}
    {\vdash \ceil{f} \colon A \li B}
    \label{lrule:type-iso}

    \infer*[lab = qinit]
    {}{\ \vdash \qinj \colon n \li n + m}
    \label{lrule:qinit}

    \infer*[lab = unitary]
    {U \colon \CC^{n} \longrightarrow \CC^{n}; \text{Unitary}}
    {\vdash \gU \colon n \li n}
    \label{lrule:unitary}

    \infer*[lab = qif]
    {\Gamma \vdash M \colon \qbit\\
      \Gamma' \vdash N_1 \colon A\\
      \Gamma' \vdash N_2 \colon A}
    {\Gamma, \Gamma' \vdash
      \qifx{M}{N_1}{N_2} \colon \qbit \after A}
    \label{lrule:qif}
  \end{proofrules}
  \vspace{-3ex}
  \caption{Typing rule of \ourlang{}.}
  \label{fig:typing-rule}
\end{figure}
Typing rules for terms derived from linear lambda calculus are standard
and can be found in the literature.

We explain the terms that are derived from BV-logic.
The two associators have the type as expected.
With the typing rule \ref{lrule:seq-intro},
we can introduce a term of $\seq$-type.
This rule, for example, allows us to derive a term of type
$A \otimes B \li A \seq B$ as
\[
  \vdash
  \lambda x. \letx{y \otimes z}{x}{(y \seq z)} \colon
  A \otimes B \li A \seq B.
\]
With \ref{lrule:let-seq}, we can destruct a term of type $A \seq B$.
This corresponds to the functoriality of $\seq$.
The rule \ref{lrule:interch} corresponds
to the interchange law in BV-logic.
Notably, from this rule, we can also derive its dual, as in
\cref{fig:derivation:term-dual-of-interch}.
\begin{figure}[t]
  \centering
  \begin{derivation}
    \infer*
    {
      \infer*{
        \infer*{
          \Gamma \vdash M \colon (A \li C) \seq (B \li D)
        }{
          \Gamma, x \colon A \seq B \vdash M \otimes x \colon 
          ((A \li C) \seq (B \li D)) \otimes (A \seq B)
        }
      }{
        \Gamma, x \colon A \seq B \vdash \sinterch(M \otimes x) \colon 
        ((A \li C) \otimes A)  \seq ((B \li D) \otimes B)
      }
      \qquad
      \infer*{
      y \colon (A \li C) \otimes A \vdash
      \Let \cdots \colon C
      \\\\
      z \colon (A \li C) \otimes A \vdash
      \Let \cdots \colon D
      }{}
    }{
      \Gamma \vdash
      \lambda x.
      \letx{y \seq z}{
        (\sinterch{M \otimes x})
      }{
          (\letx{y_1 \otimes y_2}{y}{y_1\, y_2})
          \seq
          (\letx{z_1 \otimes z_2}{z}{z_1\, z_2})
      }
      \colon 
      (A \seq B) \li (C \seq D)
    }
  \end{derivation}
  \caption{A term corresponding to the dual of the interchange rule}
  \label{fig:derivation:term-dual-of-interch}
\end{figure}

The \ref{lrule:await} and \ref{lrule:expose} rules are the rules
that are not derived from the original BV-logic,
but from the extension of adding first-order propositions.
Intuitively, with the term $\pawait$, if $A$ is first-order,
we can obtain a pair of data $A$ and $B$
from causally ordered data $A \seq B$
by ``waiting'' for the computation that produces $A$ to finish.
On the other hand, with the term $\pexpose$, again if $A$ is first-order,
we can ``expose'' the output data $A$ from a function $B \li (C \seq A)$
before it is actually executed, but with a restriction that says $A$
has to be used after the function $B \li C$ is resolved.
Note that \ref{lrule:expose} is the only structural term with
no corresponding one when we replace every occurrence of $\seq$ with $\otimes$.

In the remaining terms,
we have \ref{lrule:type-iso} for type casting %
\ref{lrule:qinit} for injection,
and \ref{lrule:unitary} for unitary maps.
The last rule \ref{lrule:qif} for quantum conditional statement is notable.
The $\qif$ term takes a $\qbit$ as an input for the condition,
and returns data of the type $\qbit \after A$,
creating the superposition of $N_1$ and $N_2$ of type $A$.
An intuition behind this typing rule is that,
since the control qubit is required in the execution of $A$,
we have to wait until $A$ is resolved to reuse the control qubit.
In fact, for example, when $A$ is given by
$B_1 \li B_2 \li \cdots B_n \li n$,
the return type $\qbit \after A$ has the following equivalence of types
\begin{align*}
\qbit \after A
&\;\;\equiv\;\;
\qbit \after (B_1 \li B_2 \li \cdots B_n \li n)
\;\;\equiv\;\;
B_1 \li (\qbit \after B_2 \li \cdots B_n \li n)
\\
\cdots&
\;\;\equiv\;\;
B_1 \li B_2 \li \cdots B_n \li (\qbit \after n)
\;\;\equiv\;\;
B_1 \li B_2 \li \cdots B_n \li (\qbit \otimes n),
\end{align*}
which means that the control qubit is returned only after
the execution of $A$ is complete.
Note that this equivalence relies heavily on the fact that
$\qbit$ is first-order.

\begin{remark}
  We defined the type of $\qif$ by $\qbit \after A$
  based on IBV ($\otimes, \li, \seq$),
  but we could instead define the type as $\qbit \llpar A$ and
  adopt classical MLL types ($\otimes, \llpar$).
  Indeed, we can have a similar type equivalence
  with this definition.
  \begin{align*}
    \qbit \llpar (B_1 \li B_2 \li \cdots B_n \li n)
    \quad\cong\quad
    B_1 \li B_2 \li \cdots B_n \li (\qbit \llpar n),
  \end{align*}
  Using \( \qbit \llpar A \) instead of \( \qbit \after A \) can be justified by a law of causal logic, \( \qbit \after A \cong \qbit \llpar A \).

  The reason why we use IBV and causal logic is as follows.
  First, classical MLL involves computational complications arising from its classical nature, whereas the intuitionistic system IBV is easier to relate to programs.
  Second, a type system based on classical MLL needs an additional isomorphism \( \qbit \llpar n \cong \qbit \otimes n \), of which justification cannot be found except for causal logic.
  Third, we also believe our explanation that ``the control qubit should live after the execution of the branches'' intuitively motivates the introduction of $\seq$.
  \lipicsEnd
\end{remark}

The next proposition
states that our language covers all derivable proofs in IBV.
A proof can be found in \cref{app:sec:logic-lang}.
\begin{proposition}
  \label{prop:program-as-proof}
  Each sequentIBV derivation has a corresponding term in \ourlang{}.
  \lipicsEnd
\end{proposition}

\begin{remark}
  Since sequentIBV does not admit cut elimination,
  our language based on it does not have a natural operational semantics.
  The language we propose is intended for writing supermaps, and given that not every supermaps inherently possess a natural operational meaning, we deliberately choose not to attempt to define an operational semantics in this paper.
  \lipicsEnd
\end{remark}

\subsection{Examples}

\begin{example}
  \label{ex:switch-term-definition}
  In our language, we can define the following term
  $\mathit{SWITCH}$
  that makes the superposition of the order of two function applications
  as the quantum SWITCH~\cite{Chiribella2013}.
  We write $g \circ f$ to mean $\lambda y.\,g\,(f\,y)$.
  \[
    x \colon \qbit,
    f, g \colon A \li A
    \ \vdash\ 
    \qifx{x}{g \circ f}{f \circ g}
    \ \colon\ 
    \qbit \after (A \li A)
  \]

  When $A = d$ for some $d \in \Nat^+$,
  we will later see that the semantics of this term
  actually defines the quantum SWITCH.
  In this case, the type equivalence above means that
  this term is essentially equivalent to the following term
  \begin{align*}
    x \colon \qbit,
    f, g \colon d \li d
    \ \vdash\
    \lambda y. \pawait(\qifx{x}{g\,(f\, y)}{f\,(g \, y)})
    \ \colon\ d \li (\qbit \otimes d)
  \end{align*}
  where the variable $y$ now appears as an input to the term.
  \lipicsEnd
\end{example}

\begin{example}
  Let us assume that we are given two functions
  $\Gamma \vdash M \colon A \li A' \otimes C$
  and
  $\Gamma' \vdash N \colon C \otimes B \li B'$,
  and we want to connect these via $C$ as in the following circuit.
  
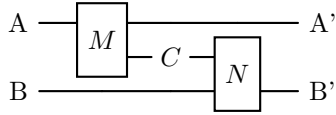
\[
    \begin{quantikz}[row sep=-0.2em]
      \lstick{A}   & \gate[2]{M} &[-5pt] &[-5pt] & \rstick{A'} \\
      & \setwiretype{n} & \ C\ \setwiretype{q} & \gate[2]{N} & \setwiretype{n} \\
      \lstick{B} & & & & \rstick{B'} \\
    \end{quantikz}
  \]
  In our language, we can describe this term with the type
  $(A \li A') \before (B \li B')$,
  reflecting the fact that information flows
  only from $\mathsf{A} = (A \li A')$ to $\mathsf{B} = (B \li B')$.
  To do so, we assume that $C$ is first-order,
  ensuring that
  no information flows backward through the output $C$ of $M$.
  Then we can derive the following terms.
  \begin{align*}
    \Gamma &\vdash \pexpose(\lambda a.\, \letx{a' \otimes c}{M\,a}{a' \seq c})
    \colon (A \li A') \seq C
    \\
    \Gamma' &\vdash \lambda c.\,\lambda b.\,N\,(c \otimes b)
    \colon C \li (B \li B')
  \end{align*}
  Let us denote these terms as
  $\Gamma \vdash M' \colon \mathsf{A} \seq C$
  and
  $\Gamma' \vdash N' \colon C \li \mathsf{B}$.
  We also have the following term
  $x \colon C \seq (C \li \mathsf{B}) \vdash L \colon \mathsf{B}$.
  \begin{align*}
    x \colon C \seq (C \li \mathsf{B})
    \vdash
    \letx{c \otimes f}{\pawait(x)}{(f\,c)}
    \colon \mathsf{B}
  \end{align*}
  Then we can derive the following
  \begin{derivation}
    \infer*{
      \infer*{
        \Gamma, \Gamma' \vdash M' \seq N'
        \colon (\mathsf{A} \seq C) \seq (C \li \mathsf{B})
      }{
        \Gamma, \Gamma' \vdash \sassocL(M' \seq N')
        \colon \mathsf{A} \seq (C \seq (C \li \mathsf{B}))
      }
    }{
      \Gamma, \Gamma' \vdash
      \letx{a \seq x}{\sassocL(M' \seq N')}{(a \seq L)}
      \colon \mathsf{A} \seq \mathsf{B}
    }
  \end{derivation}
  which has the desired type.
  \lipicsEnd
\end{example}

\begin{example}
  \label{ex:auxiliary-space}
  Let us assume we are given a supermap
  $\Gamma \vdash M \colon (A \li B) \li C$
  and a map $\Gamma' \vdash N \colon A \otimes A' \li B \otimes B'$,
  and we want to connect them via $A$ and $B$
  as in the following image
  to obtain a map of type $A' \li B' \otimes C$.
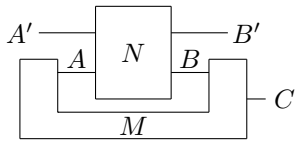
  \[
\begin{tikzpicture}[yscale=0.7]
	\begin{pgfonlayer}{nodelayer}
		\node [style=none] (0) at (-1.5, 1.25) {};
		\node [style=none] (1) at (-1.5, -0.25) {};
		\node [style=none] (2) at (-1, 1.25) {};
		\node [style=none] (3) at (-1, 0.25) {};
		\node [style=none] (4) at (1, 0.25) {};
		\node [style=none] (5) at (1, 1.25) {};
		\node [style=none] (6) at (1.5, 1.25) {};
		\node [style=none] (7) at (1.5, -0.25) {};
		\node [style=none] (8) at (-1, 1) {};
		\node [style=none] (9) at (1, 1) {};
		\node [style=none] (10) at (-0.5, 1) {};
		\node [style=none] (11) at (0.5, 1) {};
		\node [style=none] (12) at (-0.75, 1.25) {$A$};
		\node [style=none] (13) at (0.75, 1.25) {$B$};
		\node [style=none] (14) at (1.5, 0.5) {};
		\node [style=none] (15) at (1.75, 0.5) {};
		\node [style=none] (16) at (2, 0.5) {$C$};
		\node [style=none] (17) at (-0.5, 2.25) {};
		\node [style=none] (18) at (-0.5, 0.5) {};
		\node [style=none] (19) at (0.5, 0.5) {};
		\node [style=none] (20) at (0.5, 2.25) {};
		\node [style=none] (21) at (-1.25, 1.75) {};
		\node [style=none] (22) at (-0.5, 1.75) {};
		\node [style=none] (23) at (0.5, 1.75) {};
		\node [style=none] (24) at (1.25, 1.75) {};
		\node [style=none] (25) at (-1.5, 1.75) {$A'$};
		\node [style=none] (26) at (1.5, 1.75) {$B'$};
		\node [style=none] (27) at (0, 0) {$M$};
		\node [style=none] (28) at (0, 1.375) {$N$};
	\end{pgfonlayer}
	\begin{pgfonlayer}{edgelayer}
		\draw (0.center) to (1.center);
		\draw (1.center) to (7.center);
		\draw (7.center) to (6.center);
		\draw (6.center) to (5.center);
		\draw (5.center) to (4.center);
		\draw (4.center) to (3.center);
		\draw (3.center) to (2.center);
		\draw (2.center) to (0.center);
		\draw (3.center) to (3.center);
		\draw (8.center) to (8.center);
		\draw (8.center) to (10.center);
		\draw (11.center) to (9.center);
		\draw (14.center) to (15.center);
		\draw (17.center) to (20.center);
		\draw (20.center) to (19.center);
		\draw (19.center) to (18.center);
		\draw (18.center) to (17.center);
		\draw (23.center) to (24.center);
		\draw (21.center) to (22.center);
	\end{pgfonlayer}
\end{tikzpicture}
   \]
  We first assume that $B'$ is first-order.
  Then we have the following terms $M'$ and $N'$.
  \begin{align*}
    \Gamma &\vdash
    M \seq \unitval
    \colon ((A \li B) \li C) \seq \bot
    \\
    \Gamma', a' \colon A' &\vdash 
    \pexpose(\lambda a.\,
    \letx{b \otimes b'}{N(a \otimes a)}{b \seq b'})
    \colon (A \li B) \seq B'
  \end{align*}
  Now, the term $\sinterch(M' \otimes N')$ has type
  $(((A\li B) \li C) \otimes (A \li B)) \seq (I \otimes B')$.
  From the left-hand side of $\seq$, we can obtain $C$,
  and from the right-hand side, we can obtain $B'$.
  Therefore, there is a term $L$ of type
  $\Gamma, \Gamma', a' \colon A' \vdash L \colon C \seq B'$.
  By assuming that $C$ is also a first-order type,
  we can obtain the term
  $\Gamma, \Gamma' \vdash \lambda a'.\,\pawait(L)
  \colon A \li C \otimes B'$.

  In particular, when we are given another supermap of type
  $(A' \li B') \li C'$ as an input, where $C'$ is also FO,
  we can connect $A'$ and $B'$ again to obtain $C \otimes C'$
  as in
  
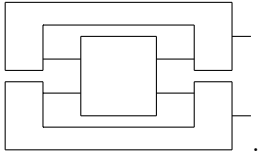
\begin{align*}
\begin{tikzpicture}[yscale=0.6]
	\begin{pgfonlayer}{nodelayer}
		\node [style=none] (0) at (-1.5, 1.25) {};
		\node [style=none] (1) at (-1.5, -0.25) {};
		\node [style=none] (2) at (-1, 1.25) {};
		\node [style=none] (3) at (-1, 0.25) {};
		\node [style=none] (4) at (1, 0.25) {};
		\node [style=none] (5) at (1, 1.25) {};
		\node [style=none] (6) at (1.5, 1.25) {};
		\node [style=none] (7) at (1.5, -0.25) {};
		\node [style=none] (8) at (-1, 1) {};
		\node [style=none] (9) at (1, 1) {};
		\node [style=none] (10) at (-0.5, 1) {};
		\node [style=none] (11) at (0.5, 1) {};
		\node [style=none] (14) at (1.5, 0.5) {};
		\node [style=none] (15) at (1.75, 0.5) {};
		\node [style=none] (17) at (-0.5, 2.25) {};
		\node [style=none] (18) at (-0.5, 0.5) {};
		\node [style=none] (19) at (0.5, 0.5) {};
		\node [style=none] (20) at (0.5, 2.25) {};
		\node [style=none] (21) at (-1, 1.75) {};
		\node [style=none] (22) at (-0.5, 1.75) {};
		\node [style=none] (23) at (0.5, 1.75) {};
		\node [style=none] (24) at (1, 1.75) {};
		\node [style=none] (25) at (-1.5, 1.75) {};
		\node [style=none] (29) at (-1.5, 1.5) {};
		\node [style=none] (30) at (-1.5, 3) {};
		\node [style=none] (31) at (-1, 1.5) {};
		\node [style=none] (32) at (-1, 2.5) {};
		\node [style=none] (33) at (1, 2.5) {};
		\node [style=none] (34) at (1, 1.5) {};
		\node [style=none] (35) at (1.5, 1.5) {};
		\node [style=none] (36) at (1.5, 3) {};
		\node [style=none] (41) at (1.5, 2.25) {};
		\node [style=none] (42) at (1.75, 2.25) {};
	\end{pgfonlayer}
	\begin{pgfonlayer}{edgelayer}
		\draw (0.center) to (1.center);
		\draw (1.center) to (7.center);
		\draw (7.center) to (6.center);
		\draw (6.center) to (5.center);
		\draw (5.center) to (4.center);
		\draw (4.center) to (3.center);
		\draw (3.center) to (2.center);
		\draw (2.center) to (0.center);
		\draw (3.center) to (3.center);
		\draw (8.center) to (8.center);
		\draw (8.center) to (10.center);
		\draw (11.center) to (9.center);
		\draw (14.center) to (15.center);
		\draw (17.center) to (20.center);
		\draw (20.center) to (19.center);
		\draw (19.center) to (18.center);
		\draw (18.center) to (17.center);
		\draw (23.center) to (24.center);
		\draw (21.center) to (22.center);
		\draw (29.center) to (30.center);
		\draw (30.center) to (36.center);
		\draw (36.center) to (35.center);
		\draw (35.center) to (34.center);
		\draw (34.center) to (33.center);
		\draw (33.center) to (32.center);
		\draw (32.center) to (31.center);
		\draw (31.center) to (29.center);
		\draw (32.center) to (32.center);
		\draw (41.center) to (42.center);
	\end{pgfonlayer}
\end{tikzpicture}
.
    \tag*{\lipicsEnd}
  \end{align*}
\end{example}

\begin{example}
  In our language, we can define the following map
  \begin{align*}
    & x, y \colon \qbit,
    f, g \colon \qbit \li \qbit
    \\
    &
    \quad\vdash
      \qifx{x}
      {\mathit{SWITCH}(y, f, g)}
      {(f\,y) \after g}
    \ \colon\ 
    \qbit \after (\qbit \li \qbit),
  \end{align*}
  where $\mathit{SWITCH}$ is the term defined in \cref{ex:switch-term-definition}.
  The qubit $y$
  is used as a control in the \texttt{then} branch
  and as an input of $f$ in the \texttt{else} branch.
  Therefore, this function cannot be represented
  in a language that separates
  ``control'' qubits and ``target'' qubits
  as in \cite{ClmPerd20-LIPIcs-PBS}.
  \lipicsEnd
\end{example}

\subsection{Degenerate Categorical Semantics}
\label{subsec:lang-degenerate-semantics}

A degenerate categorical semantics for our language can be given in $\Hilb$,
using the compact closed structure where
$n \otimes m = n \seq m = n \li m = n \times m$.

For a type $A$, its categorical semantics $\sem{A}_\Hilb \in \Hilb$
is inductively defined as follows:
\[
  \sem{n} = n,\quad
  \sem{\bot} = 1,\quad
  \sem{A \otimes B} = \sem{A \seq B} = \sem{A \li B} = \sem{A} \otimes \sem{B}.
\]

The semantics for terms $\Gamma \vdash M \colon A$
is given in \cref{fig:degenerate-cat-sem}
as a morphism $\sem{M}_\Hilb \colon \sem{\Gamma} \to \sem{A}$ in $\Hilb$
where $\sem{\Gamma}$ represents the tensor product of
the semantics of types in $\Gamma$.
\begin{figure*}
  \begin{proofrules*}
    \sem{x \colon A \vdash x \colon A} = \id_{\sem{A}}

    \sem{\unitval} = \id_{I}

    \sem{\lambda x.\,M} = \Lambda_{\sem{\Gamma},\sem{A},\sem{B}} \sem{M}

    \sem{M\,N} =  \mathrm{ev}_{\sem{A},\sem{B}} \circ (\sem{M} \otimes \sem{N})

    \sem{M \otimes N} =
    \sem{M \seq N} = 
    \sem{M} \otimes \sem{N}

    \sem{\letx{\,\_}{M}{N}} = \sem{M} \otimes \sem{N}

    \sem{\letx{x \otimes y}{M}{N}}
    = \sem{N} \circ (\sem{M} \otimes \id_{\sem{\Gamma'}})

    \sem{\letx{x \seq y}{M}{(N_1 \seq N_2)}}
    = (\sem{N_1} \otimes \sem{N_2})
    \sigma (\sem{M} \otimes \id_{\sem{\Gamma_1}} \otimes \id_{\sem{\Gamma_2}})

    \sem{\sassocR}
    = \alpha

    \sem{\sassocL}
    = \alpha^{-1}

    \sem{\sinterch(M)}
    = \sigma \circ \sem{M}

    \sem{\pawait(M)}
    = \sem{M}

    \sem{\pexpose(M)}
    = \sem{M}

    \sem{\qinj}
    = \qinj

    \sem{\gU}
    = U

    \sem{\qifx{M}{N_1}{N_2}}
    = \qifmap_{A} \circ (\sem{M} \otimes \langle\sem{N_1}, \sem{N_2}\rangle)
  \end{proofrules*}
  \caption{Degenerate categorical semantics in $\Hilb$.}
  \vspace{-0.7em}
  \captionsetup{singlelinecheck=off}
  \caption*{
    Here,
    $\Lambda_{X,Y,Z} F \colon X \to (Y \li Z)$ is the transpose map of
    $F \colon X \otimes Y \to Z$,\ \
    $\mathrm{ev}_{X, Y} \colon (X \li Y) \otimes X \to Y$ is the evaluation map,\ \
    $\sigma \colon
    X \otimes Y \otimes Z \otimes W \to
    X \otimes Z \otimes Y \otimes W$ is the isomorphism that swaps the middle two,\ \
    $\langle f, g \rangle \colon X \to Y \otimes Z$ is the map 
    that is obtained from the universality of the product
    from $f \colon X \to Y$ and $g \colon X \to Z$.
  }
  \label{fig:degenerate-cat-sem}
\end{figure*}
For most of the terms, the semantics are based on standard linear lambda calculus.
For example, the semantics of lambda abstraction and function application
are defined via the monoidal closed structure of $\Hilb$.
As already noted,
by identifying $\seq$ and $\otimes$,
all the terms in the first three lines of~\cref{fig:degenerate-cat-sem}
can be regarded as those in a standard linear lambda calculus.
For example, the term
$\letx{x \seq y}{M}{(N_1 \seq N_2)}$
is identified with
$\letx{x \otimes y}{M}{(N_1 \otimes N_2)}$,
so they are defined to have the same semantics.

In the last line of~\cref{fig:degenerate-cat-sem},
we have the definition of semantics of $\pawait$, $\pexpose$,
and the three quantum primitives.
For $\pawait$ and $\pexpose$, the domain and the codomain of the semantics
in $\Hilb$ match, so they can be simply defined as the identities.
The semantics of $\qinj$ and $U$ are defined by the followings.
\begin{align*}
  \qinj\ \colon\ d \longrightarrow d + d';\ \ket i \longmapsto \ket i,
  \qquad
  U\ \colon\ d \longrightarrow d;\ \ket\psi \longmapsto U \ket\psi
  .
\end{align*}
For the semantics of $\qifx{M}{N_1}{N_2}$,
we use a linear map
$\qifmap_n \colon 2 \otimes (n \oplus_\Hilb n) \longrightarrow 2 \otimes n$
for each object $n$ as follows:
\begin{align*}
  \begin{matrix}
    \qifmap_{n} \colon
     & 2 \otimes (n \oplus n)
     & \longrightarrow
     & 2 \otimes n
    \\
     & \ket{0} \otimes (\ket{\varphi} \oplus \ket{\psi})
     & \longmapsto
     & \ket{0} \otimes \ket{\varphi}
     \ \
    \\
     & \ket{1} \otimes (\ket{\varphi} \oplus \ket{\psi})
     & \longmapsto
     & \ket{1} \otimes \ket{\psi}.
  \end{matrix}
\end{align*}

The following two examples show that
the semantics of quantum conditional branching is defined as expected
with this map $\qifmap_n$.

\begin{example}
  Let $M$ be
  $\qifx{x}{\gU\,y}{y}$.
  This term has the type derivation $x, y \colon \qbit \vdash M \colon \qbit \after \qbit$,
  and the semantics are given by the controlled-$U$ gate as
  $\sem{M} = CU$.
  \lipicsEnd
\end{example}

\begin{example}
  When $A \in \Nat^+$,
  the semantics of the term
  $x \colon \qbit,\allowbreak f,g \colon A \li A
  \vdash \mathit{SWITCH} \colon\allowbreak \qbit \after (A \li A)$
  we defined in \cref{ex:switch-term-definition}
  coincides with the usual quantum SWITCH~\cite{Chiribella2013}
  described as
  \begin{align*}
    (\alpha\ket0 + \beta\ket1) \otimes f \otimes g
    \longmapsto
    (\alpha\ket0 \otimes (g f))
    +
    (\beta\ket1 \otimes (f g)).
    \tag*{\lipicsEnd}
  \end{align*}
\end{example}

\begin{remark}
  It is worth noting that a closed monoidal category is a (degenerate) model of IBV, in which \( (A \before B) \cong (A \otimes B) \).
  This is in contract to the classical setting: a \( * \)-autonomous category is not necessarily a model of BV logic.
  The difference comes from the self-duality of \( \before \) required only in the classical case.
  A model of (classical) BV logic satisfying \( (A \before B) \cong (A \otimes B) \) also satisfies \( (A \otimes B) \cong (A \before B) \cong \neg((\neg A) \before (\neg B)) \cong \neg ((\neg A) \otimes (\neg B)) \cong (A \llpar B) \), so it is compact closed.
  We think that this difference in the behavior of \( \before \) between the intuitionistic and classical settings is interesting, although we cannot discuss it further in this paper.
  \lipicsEnd
\end{remark}

\section{Causal Model of Pure Quantum Computation}
\label{sec:CausHilb}

Although we have already defined a categorical model
of \ourlang{} in \cref{subsec:lang-degenerate-semantics},
it does not reflect the causal structure of the language well enough
since it is defined in $\Hilb$ where $\otimes = \seq = \llpar$.
This is problematic and does not allow us to analyze our language
in detail.
For example, it is not easy to prove whether
there is no program that increases the total probability of a state
as in the paradox we explained in the Introduction.

In this section, we construct a new categorical model
of pure quantum computation $\CausHilb$,
which has enough structure to interpret BV-logic
and our language.

\subsection{Caus Construction}

Kissinger, Uijlen and Simmons~\cite{Kissinger2019,Simmons2022}
have developed a construction of a categorical model of causality,
called the \emph{Caus construction}.
With the Caus construction, from a compact closed category $\categoryC$
with products and \emph{discarding maps},
called an \emph{additive precausal category},
one can construct a category $\Caus[\categoryC]$
that captures the essence of causal ordering of processes.
Here, we briefly review their results.

\begin{definition}
  A compact closed category $\categoryC$ with products
  is called \emph{additive precausal}
  if it is equipped with a family of morphisms
  $\{ \discardDiag_A \in \categoryC(A, I) \}_{A \in \categoryC}$
  satisfying the following axioms.
  \begin{enumerate}[
      label=AP\arabic*.,
      leftmargin=\dimexpr 26pt-.0em
    ]
    \item
          $\discardDiag_{A \otimes B} = \discardDiag_A \otimes \discardDiag_B$,\
          $\discardDiag_{I} = \id_I$,\
          $\discardDiag_{A \oplus B} = \discardDiag_A \oplus \discardDiag_B$.
    \item $\discardDiag_A \circ \maxmixDiag_A$ is invertible for any $A \neq 0$,
          where $\maxmixDiag_A \defeq (\discardDiag_{A^*})^*$.
    \item For each $A$, there exists a finite set
          $\{ \rho_i \in \categoryC(I, A) \}_{i = 1}^{n}$
          such that $\discardDiag_A \circ \rho_i = \id_I$ and
          jointly monic, \ie
          $\forall f, g \in \categoryC(A, B).\
            (\forall i.\ f \circ \rho_i = g \circ \rho_i)
            \implies f = g$.
    \item For the commutative semiring of scalars $R \defeq \categoryC(I, I)$,
          the addition is cancellative,
          the canonical preorder induced by addition is total,
          and $R^\times$ forms a group by multiplication.\label{axiom:addtive-precausal:scalar}
    \item $\forall \pi \in \categoryC(A, I).\
            \exists \pi' \in \categoryC(A, I).\
            \exists \lambda \in R.\
            \pi + \pi' = \lambda \cdot \discardDiag_A
          $.
          \lipicsEnd
  \end{enumerate}
\end{definition}

\begin{definition}
  Let $\categoryC$ be an additive precausal category.
  For a set $c \subseteq \categoryC(I, A)$,
  we define the set $c^* \subseteq \categoryC(A, I)$ by
  $\{ \pi \in \categoryC(A, I)
    \mid \forall \rho \in c.\ \pi \circ \rho = \id_I \}$.
  If $c^{**} = c$ holds, $c$ is called \emph{closed}.

  A set $c \subseteq \categoryC(I, A)$ is called \emph{flat}
  if there exist invertible scalars $\mu$ and $\theta$
  such that $\mu \cdot \maxmixDiag_A \in c$
  and $\theta \cdot \discardDiag_A \in c^*$.

  The category $\Caus[\categoryC]$ consists of the following:
  An object is a pair $\mathbf{A} = (A, c_\mathbf{A})$ of an object $A$
  and a closed and flat set $c_\mathbf{A} \subseteq \categoryC(I, A)$.
  A morphism $f \colon \mathbf{A} \longrightarrow \mathbf{B}$ is
  a map $f \in \categoryC(A, B)$ such that
  $\forall \rho \in c_\mathbf{A}.\ f \circ \rho \in c_B$.
  \lipicsEnd
\end{definition}

The category $\Caus[\categoryC]$ is a full subcategory
of $\mathbf{T}(\categoryC)$
where the objects are pairs $(A, c_\mathbf{A})$
of an object $A$ and a closed set $c_\mathbf{A}$.
In \cite{HYLAND2003183},
this category $\mathbf{T}(\categoryC)$ is called
the tight category induced by focussed orthogonality
for the singleton set $\{ \id_I \} \subseteq R$,
so it is $\ast$-autonomous.
By proving flatness is closed under
all $\ast$-autonomous structures,
one can prove the following.

\begin{proposition}[\cite{Kissinger2019}]
  $\Caus[\categoryC]$ is a $\ast$-autonomous category
  with products,
  where the monoidal structure and dual are defined as follows.
  The obvious forgetful functor
  $\Caus[\categoryC] \longrightarrow \categoryC$
  preserves all structures on the nose.
  \begin{align*}
    \mathbf{I}
     &
    \defeq
    (I, \{\id_I\})
    \hspace{4em}
    \mathbf{A}^*
    \defeq
    (A, c_\mathbf{A}^*)
    \\
    \mathbf{A} \otimes \mathbf{B}
     &
    \defeq
    (
    A \otimes B,
    c_{A \otimes B}
    \defeq
    \{
    \rho_A \otimes \rho_B
    \mid
    \rho_A \in c_\mathbf{A},
    \rho_B \in c_\mathbf{B}
    \}^{**}
    )
    \tag*{ \lipicsEnd }
  \end{align*}
\end{proposition}

\begin{example}
  $\CPM$ is additive precausal,
  thus it admits the $\Caus$ construction.
  \lipicsEnd
\end{example}

\subsection{Our Categorical Model: CausHilb}

The Caus construction is indeed a remarkable construction.
It is not only $\ast$-autonomous,
but can also interpret derivations in BV-logic.
Moreover, $\CausCPM$ does provide a categorical model
that captures causal relationships in quantum computing.
However, it does not provide a model for our language.
In particular, the semantics of quantum conditional $\qif$
in our language is defined using the additive enrichment in $\Hilb$
corresponding to quantum superposition,
which is not in $\CPM$ or $\CausCPM$, where the additive enrichment
corresponds to probabilistic mixing.

One might then consider applying the Caus construction
to $\Hilb$, but this does not work since $\Hilb$ is not
additive precausal.
In fact, the requirement for a category to be additive precausal
is quite strong.
For example, \ref{axiom:addtive-precausal:scalar}
asks the scalars to be totally ordered like $\Real^+$,
which is not the case for $\Hilb(1,1) = \CC$.

Instead of directly applying the Caus construction
to $\Hilb$ or $\CPM$,
we define a category by gluing together
the structure in $\Hilb$ and $\CausCPM$ as follows.

\begin{definition}
  The category $\CausHilb$ consists of the following:
  Objects of $\CausHilb$ are given by pairs $(n, c)$ of
  a natural number $n$ and a closed and flat set
  $c \subseteq \CPM(1, (n))$.
  Morphisms from $(n, c)$ to $(m, c')$ are linear maps
  $f \in \Hilb(n, m)$ such that
  $\forall \rho \in c, (\iota f) \circ \rho \in c'$.

  There is a canonical forgetful functor
  $\CausHilb \longrightarrow \Hilb$,
  and a functor
  $\CausHilb \longrightarrow \CausCPM$
  which maps $(n, c)$ to $((n), c)$
  and $f$ to $\iota(f)$.\footnote{
    Abusing the notation, we also write the functor
    $\CausHilb \longrightarrow \CausCPM$ as $\iota$.
  }
  \lipicsEnd
\end{definition}

In fact, there is an alternative definition for $\CausHilb$.

\begin{lemma}
  The following diagram is a pullback square in $\Cat$.
  \begingroup
  \makeatletter
  \def\tagform@#1{%
    \maketag@@@{(\ignorespaces#1\unskip\@@italiccorr)\quad$\lrcorner$}%
  }%
  \makeatother
  \begin{equation}
    \begin{tikzpicture}[scale=1, transform shape, baseline={(0,0.2)}]
      \node(lu) at (0,1) {$\CausHilb$};
      \node(ld) at (0,0) {$\Hilb$};
      \node(ru) at (3,1) {$\CausCPM$};
      \node(rd) at (3,0) {$\CPM$};
      \draw[1cell] (lu) -- (ld);
      \draw[1cell] (lu) -- (ru);
      \draw[1cell] (ld) -- node[] {$\iota$} (rd);
      \draw[1cell] (ru) -- node[] {$U$} (rd);
      \pullback{0.5,0.6}
    \end{tikzpicture}
    \label{diagram:pullback}
  \end{equation}
  \endgroup
\end{lemma}

\begin{proposition}
  \label{prop:causHilb-is-star-aut}
  $\CausHilb$ is $\ast$-autonomous and
  the functors $\CausHilb \longrightarrow \Hilb$
  and $\CausHilb \longrightarrow \CausCPM$
  preserves all $\ast$-autonomous structures on the nose.
\end{proposition}
\begin{proof}
  The category $\caty{AutCat}$
  of $\ast$-autonomous categories and functors
  that preserve every structure on the nose
  is monadic over $\Cat$~\cite{BLACKWELL19891}.
  So the forgetful functor $\caty{AutCat} \longrightarrow \Cat$
  creates limits.
  Since the cospan in the diagram \eqref{diagram:pullback}
  is a cospan in $\caty{AutCat}$,
  $\CausHilb$ is also $\ast$-autonomous.
\end{proof}

From the previous proposition,
we can easily compute the $\ast$-autonomous structure in $\CausHilb$.
For example, the monoidal unit is $\mathbf{I} \defeq (1, \{ \id_I \})$,
and the monoidal product of objects can be
computed in $\CausCPM$,
\ie an object $(n, c_{\mathbf{A}}) \in \CausHilb$ can be identified with
an object of $((n), c_\mathbf{A}) \in \CausCPM$,
and the monoidal product can be computed via this identification.
Note that, even though we can compute the structure in such a way,
$\CausHilb \longrightarrow \CausCPM$ is not full or faithful
and the categories themselves are very different.

\begin{proposition}
  \label{prop:caushilb-product}
  $\CausHilb$ has all finite products and coproducts,
  preserved by the forgetful functor
  $\CausHilb \longrightarrow \Hilb$.
\end{proposition}
\begin{proof}
  Let 
  $\mathbf{A} = (A, c_\mathbf{A})$
  and
  $\mathbf{B} = (B, c_\mathbf{B})$
  be objects of $\CausHilb$.
  Without loss of generality,
  we assume $\maxmixDiag_A \in c_\mathbf{A}$
  and $\maxmixDiag_B \in c_\mathbf{B}$.
  The product of the objects
  is given by
  \[
    \mathbf{A} \times \mathbf{B} \defeq
    (A \oplus_\Hilb B, \{
    \rho \mid (\iota p_A)\rho \in c_\mathbf{A}, (\iota p_B)\rho \in c_\mathbf{B}
    \}),
  \]
  where $p_\mathbf{A}$ is the projection
  $A \oplus_\Hilb B \longrightarrow A$ in $\Hilb$.
  Since $\CausHilb$ is self-dual, it also admits coproducts
  $\mathbf{A} \sqcup \mathbf{B} \defeq
    (\mathbf{A}^* \times \mathbf{B}^*)^*$.
  A detailed proof can be found in \cref{app:sec:product}.
\end{proof}

\subsection{BV-category}
\label{subsec:bv-category}

We now show that our model forms a BV-category,
in which BV-logic has interpretation.
BV-category was first defined in~\cite{Blute2012},
and later studied in \cite{BV-model-HeffordWilson_2025}
in detail.
We first sketch the definition of BV-categories,
asking the reader to consult the other papers for details.

\begin{definition}[\cite{Blute2012,BV-model-HeffordWilson_2025}]
  A BV-category is a category with three monoidal structures
  $\otimes$, $\llpar$, and $\seq$ with isomorphic units $I \iso I^* \iso J$,
  with the following structures:
  \begin{itemize}
    \item a $\ast$-autonomous isoMIX structure with respect to $\otimes$ and $\llpar$,
    \item a self-dual structure $(A \seq B)^* \iso A^* \seq B^*$,
    \item a duoidal category structure with respect to $\otimes$ and $\seq$,
          \ie interchange natural transformation $
            \zeta\colon(A\seq C)\otimes(B\seq D) \longrightarrow
            (A \otimes B) \seq (C \otimes D)
          $
          that is compatible with the monoidal structures,
    \item and another duoidal category structure
          with respect to $\seq$ and $\llpar$,
          which derives as the dual of that for $\otimes$ and $\seq$.
          \lipicsEnd
  \end{itemize}
\end{definition}

To distinguish three associators and left/right unitors for BV-categories,
we call them $\alpha^\bullet$, $\ell^\bullet$ and $r^\bullet$
for each $\bullet \in \{\otimes,\llpar,\seq\}$.

\begin{remark}
  Although the BV-categories are expected to serve as a sound model of BV-logic in a certain appropriate sense,
  this remains a partially open problem.
  Recently, Acclavio~et~al.~\cite{AcclavioSZ26} have proven that a
  \emph{strong BV-category},
  which is a special BV-category with a faithful embedding into a compact closed category that satisfies some coherence axioms,
  does model BV-logic in an appropriate sense.
  Our model and the forgetful functor
  $\CausHilb \longrightarrow \Hilb$ defines a strong BV-category,
  so it is a model of BV-category.
  \lipicsEnd
\end{remark}

Rather than directly proving that $\CausHilb$ is a BV-category by spelling out all coherence diagrams and risking overlooking something, we shall prove this by showing that pullbacks in $\Cat$ preserve the BV-category structure.
For this purpose, we first state that the other three categories
$\Hilb$, $\CPM$, and $\CausCPM$
are BV-categories.

\begin{proposition}
  A compact closed category is a BV-category
  where $\otimes = \seq = \llpar$.
  \lipicsEnd
\end{proposition}

\begin{theorem}[\cite{Simmons2022}]
  $\Caus[\categoryC]$ is a BV-category
  and all structures are preserved on the nose
  by the forgetful functor
  $\Caus[\categoryC] \longrightarrow \categoryC$.
  The object $\mathbf{A} \seq \mathbf{B}$ is defined by
  $(A \otimes B, c_{\mathbf{A} \seq \mathbf{B}})$
  where
  \[
    c_{\mathbf{A} \seq \mathbf{B}} \defeq
    \big\{
    h \in \categoryC(I, A \otimes B)
    \mathrel{\big|}
    \forall \pi, \pi' \in c^*_{\mathbf{B}},
    (\id_A \otimes \pi) \circ h =
    (\id_A \otimes \pi') \circ h
    \in c_\mathbf{A}
    \big\}
    \tag*{\lipicsEnd}
  \]
\end{theorem}

To prove that pullbacks preserve BV-category structures,
we use a consequence of categorical universal algebra.

\begin{lemma}
  \label{lem:bv-monadic-over-cat}
  The category of BV-categories
  and functors that strictly preserve every structure
  on the nose is finitary monadic over $\Cat$.
\end{lemma}
\begin{proof}[Sketch of proof]
  The theory of categories $\mathbb{T}_{\mathrm{Cat}}$ can be given as
  a partial Horn theory~\cite{PalmgrenV07}
  of two sorts $C_0$ (objects) and $C_1$ (morphisms)
  with function symbols such as $\mathrm{dom}\colon C_1 \to C_0$.
  The theory of BV-categories can be given by adding
  more function symbols such as
  $\seq_0 \colon C_0, C_0 \to C_0$
  (the action of the functor $\seq$ to objects)
  or
  $\zeta \colon C_0, C_0, C_0, C_0 \to C_1$
  (the natural transformation $\zeta$)
  and many equations corresponding to the coherence diagrams.
  Therefore, the theory of BV-categories can be written
  as a $\mathbb{T}_{\mathrm{Cat}}$-relative algebraic theory in \cite{kawase2026},
  so the category of BV-categories is monadic over $\Cat$.
\end{proof}

\begin{corollary}
  \label{thm:causHilb-is-bv}
  $\CausHilb$ is a BV-category,
  and all structures are preserved on the nose
  by the functors $\CausHilb \longrightarrow \Hilb$
  and $\CausHilb \longrightarrow \CausCPM$.
\end{corollary}
\begin{proof}
  From \cref{lem:bv-monadic-over-cat},
  the forgetful functor creates the pullback \eqref{diagram:pullback} in $\Cat$.
\end{proof}

The BV-category structure of $\CausHilb$ on objects
can also be computed via the identification through $\EmbeddingFunctor$.
That is, since $\EmbeddingFunctor$ is injective on objects,
$\mathbf{A} \seq \mathbf{B}$ in $\CausHilb$
can be computed by computing
$(\iota\mathbf{A}) \seq (\iota\mathbf{B})$ in $\CausCPM$.
On the other hand, for any morphisms $f \in \CausHilb(\mathbf{A}, \mathbf{B})$
and $g \in \CausHilb(\mathbf{A'}, \mathbf{B'})$,
since the forgetful functor to $\Hilb$ preserves the structures,
the underlying morphisms of $f \otimes g$, $f \llpar g$, and $f \seq g$
in $\CausHilb$ are all $f \otimes g \in \Hilb(A \otimes A', B \otimes B')$.

\subsection{First-order Objects in CausHilb}
\label{subsec:fo-obj-in-caushilb}

In a category constructed
by the $\Caus$ construction~\cite{Kissinger2019,Simmons2022,Simmons2024},
it has been shown that the objects called \emph{first-order}
are particularly crucial.
These are the objects that can be written as
$(A, \{\discardDiag_A\}^*) \in \Caus[\categoryC]$.
Among all objects $(A, c_\mathbf{A})$
whose underlying object is $A \in \categoryC$,
such first-order objects are the ones that have
the maximal set $c_\mathbf{A} \subseteq \categoryC(I, A)$.
Since an intuitive meaning of $c_\mathbf{A}$
is the collection of causally coherent states,
first-order objects are the most unrestricted, simplest objects.
For example, in the category $\CausCPM$,
the full subcategory of first-order objects defines the category $\mathbf{CPTP}$,
the category of completely positive trace preserving maps,
which is a natural model of first-order quantum computation.

The most important aspect of first-order objects is that,
it models the first-order propositions
in the causal logic.
Indeed, in~\cite{Simmons2024}, Simmons and Kissinger have shown that
$\Caus[\categoryC]$ defines a complete model of causal logic
for any non-trivial $\categoryC$.

In our model $\CausHilb$, in addition to the BV-category structure,
we can also inherit first-order objects from $\CausCPM$.
We also see that these objects indeed capture
first-order pure quantum computation.

\begin{definition}
  An object $(n, c_\mathbf{A})$ in $\CausHilb$ is called \emph{first-order}
  if $c_\mathbf{A} = \{\discardDiag_{(n)}\}^*$.\footnote{
    Instead of this definition, we can also define
    first-order objects as objects whose $c^*_\mathbf{A}$ is a singleton.
    Since such objects are isomorphic to one of $(n, \{\discardDiag\}^*)$,
    these two definitions define an equivalent full subcategory of $\CausHilb$.
  }
  \lipicsEnd
\end{definition}

\begin{theorem}
  \label{thm:isom-is-full-subcategory}
  The full subcategory of first-order objects in $\CausHilb$
  is isomorphic to
  the category $\Isom$ of isometries,
  \ie the category whose objects are natural numbers
  and whose morphisms $n \longrightarrow m$ are isometries
  $\CC^n \longrightarrow \CC^m$.
\end{theorem}
\begin{proof}
  Let us denote the first-order object $(n, \{\discardDiag_{(n)}\}^*)$
  as $\mathbf{n}$.
  A morphism $\mathbf{n} \longrightarrow \mathbf{m}$ in $\CausHilb$
  is a linear map $\ket\psi \longmapsto V\ket\psi$
  such that $\rho\longmapsto V\rho V^\dagger$
  defines a trace-preserving map.
  Such $V$ is an isometry.
\end{proof}

\begin{remark}
  As in $\Caus[\categoryC]$, first-order objects are closed under $\otimes$,
  and there, all three monoidal products coincide,
  \ie $\mathbf{n} \otimes \mathbf{m} = \mathbf{n} \seq \mathbf{m}
    = \mathbf{n} \llpar \mathbf{m}$.
  Also, the dual of a first-order object is first-order if and only if
  it is $\mathbf{I}$.

  On the other hand, in $\Caus[\categoryC]$, it is shown that
  first-order objects are closed under coproducts, but this is not true in $\CausHilb$.
  In our model $\CausHilb$, first-order objects are closed under neither products nor coproducts.
  Thus, for example, the semantics of the $\qbit$ type,
  which will be defined to be the first-order object $\mathbf{2}$
  in the next section,
  is not a product or coproduct of the units $\mathbf{I}$.
  \lipicsEnd
\end{remark}

We show that $\CausHilb$ can also model first-order propositions
by these first-order objects.
The key observation here is that
the embedding functor $\CausHilb \longrightarrow \CausCPM$ is conservative, \ie it reflects isomorphisms.
From this fact, we can state the following: for any proof $A \vdash B$ in BV-logic,
we can give a semantics of the proof in $\CausHilb$ and $\CausCPM$
by specifying an object in $\CausHilb$ for each atomic proposition that occurs in $A$ or $B$.
These semantics are related by the embedding functor.
Then, if the semantics in $\CausCPM$ happened to be an isomorphism,
the semantics of the proof in $\CausHilb$ must also be an isomorphism.
This idea can be used in the next proposition.

\begin{proposition}
  \label{prop:first-order-prop-invertible}
  In the category $\CausHilb$, we have the following natural transformation
  as part of the duoidal structure $(\otimes,\seq)$.
  \[
    \zeta_{A,B,C,D}\colon
    (A\seq C)\otimes(B\seq D) \longrightarrow
    (A \otimes B) \seq (C \otimes D).
  \]
  From this $\zeta$, by assuming $B = C = \mathbf{I}$,
  we obtain the following $\xi_{A,B}$.
  Also, by considering $\zeta^*_{B,C^*,I,A^*}$, we obtain the following $\gamma_{B,C,A}$.
  \begin{align*}
    \xi_{A,B} \colon A \otimes B  \longrightarrow  A \seq B
    \qquad\quad
    \gamma_{B,C,A} \colon (B \li C) \seq A \longrightarrow B \li (C \seq A)
  \end{align*}
  When $A$ is first-order,
  these morphisms are isomorphisms.
\end{proposition}
\begin{proof}
  The converse $A\seq B \vdash A \otimes B$ is
  derivable in the causal logic when $A$ is a first-order proposition,
  and the semantics of the proof gives the inverse of
  $\EmbeddingFunctor(\xi_{A,B})$.
  Since $\EmbeddingFunctor (\xi_{A,B})$ is an isomorphism in $\CausCPM$,
  $\xi_{A,B}$ is also an isomorphism.
  The same proof applies to $\gamma$.
\end{proof}

\begin{remark}
  We can also prove this proposition easily by directly checking
  the existence of the inverse, but our approach is much easier to
  be generalized to any other such morphisms.
  \lipicsEnd
\end{remark}

\section{Categorical Semantics with Causal Structure}
\label{sec:cat-sem-causHilb}

We now give the semantics of \ourlang{} in $\CausHilb$
and study its properties.
We prove the full abstraction theorem in \cref{subsec:full-abst},
and we study the definability in \cref{subsec:definability}.

\subsection{Definition of Categorical Semantics}
\label{subsec:cat-sem-in-causHilb}

In this section, we define a refined categorical semantics $\sem{-}_{\CausHilb}$
of \ourlang{} in $\CausHilb$.
This categorical semantics has additional information about causality,
and by forgetting it, we can recover the original degenerate semantics
which we defined in~\cref{subsec:lang-degenerate-semantics}.
More concretely, through the forgetful functor $U: \CausHilb \longrightarrow \Hilb$,
we can write the semantics in $\Hilb$ as
$U \sem{-}_{\CausHilb} = \sem{-}_{\Hilb}$.

The semantics of types are given by the following,
relating each syntactic connective
with its corresponding categorical structure.
\begin{gather*}
  \sem{n} = (n, \{\discardDiag_{(n)}\}^*)
  \qquad
  \sem{\bot} = \mathbf{I} = (1, \{\id_{(1)}\})
  \qquad
  \sem{A \otimes B} = \sem{A} \otimes \sem{B}
  \\
  \sem{A \seq B} = \sem{A} \seq \sem{B}
  \qquad
  \sem{A \li B} = \sem{A} \li \sem{B}
\end{gather*}
For every type, the underlying object in $\Hilb$,
\ie the natural number representing the dimension,
coincides with the semantics in $\Hilb$.
It is notable that, 
the semantics of first-order types are given by first-order objects.

\begin{figure*}
  \begin{proofrules*}
    \sem{x \colon A} = \id_{\sem{A}}

    \sem{\unitval} = \id_{\mathbf{I}}

    \sem{\lambda x.\, M} = \Lambda_{\sem{\Gamma},\sem{A},\sem{B}}\sem{M}

    \sem{M\, N} = \mathrm{ev}_{\sem{A},\sem{B}} (\sem{M} \otimes \sem{N})

    \sem{M \otimes N} = \sem{M} \otimes \sem{N}

    \sem{\letx{\,\_}{M}{N}} = \ell^{\otimes} (\sem{M} \otimes \sem{N})

    \sem{\letx{x \otimes y}{M}{N}}
    = \sem{N} \circ (\sem{M} \otimes \id_{\sem{\Gamma'}})

    \sem{M \seq N}
    = \xi_{\sem{A},\sem{B}} \circ (\sem{M} \otimes \sem{N})

    \sem{\letx{x \seq y}{M}{(N_1 \seq N_2)}}
    = (\sem{N_1} \seq \sem{N_2})
    \circ
    \zeta_{\sem{A},\sem{B},\sem{\Gamma_1},\sem{\Gamma_2}}
    \circ
    (\sem{M} \otimes \xi_{\sem{\Gamma_1}, \sem{\Gamma_2}})

    \sem{\sassocR}
    = \alpha^\seq

    \sem{\sassocL}
    = {(\alpha^\seq)}^{-1}

    \sem{\sinterch(M)}
    = \zeta_{\sem{A},\sem{B},\sem{C},\sem{D}} \circ \sem{M}

    \sem{\pawait(M)}
    = \xi_{\sem{A},\sem{B}}^{-1} \circ \sem{M}

    \sem{\pexpose(M)}
    = \gamma_{\sem{B},\sem{C},\sem{A}}^{-1} \circ \sem{M}

    \sem{\qinj}
    = \qinj

    \sem{\gU}
    = U

    \sem{\qifx{M}{N_1}{N_2}}
    = \qifmap_{A} \circ (\sem{M} \otimes \langle\sem{N_1}, \sem{N_2}\rangle)
  \end{proofrules*}
  \vspace{-5ex}
  \caption{Categorical semantics of terms in $\CausHilb$}
  \label{fig:cat-sem-term}
\end{figure*}
For each term $\Gamma \vdash M \colon A$,
its categorical semantics is defined
in \cref{fig:cat-sem-term}
as a map
$\sem{M}_\CausHilb \colon \sem{\Gamma} \longrightarrow \sem{A}$
in $\CausHilb$.

Similarly to the semantics in~\cref{subsec:lang-degenerate-semantics},
the semantics of the terms in linear lambda calculus
are defined in a standard way in the monoidal closed category $\CausHilb$.
We explain the semantics for the following non-trivial ones.

For the term $\Gamma, \Gamma' \vdash M \seq N \colon A \seq B$,
the semantics is given by a morphism
$\sem{\Gamma} \otimes \sem{\Gamma'} \longrightarrow \sem{A} \seq \sem{B}$.
To obtain this morphism from $\sem{M}$ and $\sem{N}$,
we can first take the $\otimes$-monoidal product of them to obtain
a map $\sem{\Gamma} \otimes \sem{\Gamma'} \longrightarrow \sem{A} \otimes \sem{B}$,
and then post-compose the map
$\xi_{\sem{A},\sem{B}} \colon \sem{A} \otimes \sem{B}
  \longrightarrow \sem{A}\seq\sem{B}$.

For the term
$\Gamma_0, \Gamma_1, \Gamma_2 \vdash \letx{x \seq y}{M}{(N_1 \seq N_2)}$,
the semantics is given by a morphism of type
$\sem{\Gamma_0} \otimes \sem{\Gamma_1} \otimes \sem{\Gamma_2}
  \longrightarrow \sem{C} \seq \sem{D}$.
This can be obtained as follows:
\begin{align*}
  \sem{\Gamma_0} \otimes (\sem{\Gamma_1} \otimes \sem{\Gamma_2})
   &
  \xlongrightarrow{\sem{M} \otimes \xi_{\sem{\Gamma_1}, \sem{\Gamma_2}}}
  (\sem{A} \seq \sem{B}) \otimes (\sem{\Gamma_1} \seq \sem{\Gamma_2})
  \\ &
  \xlongrightarrow{\zeta_{\sem{A},\sem{\Gamma_1},\sem{B},\sem{\Gamma_2}}}
  (\sem{A} \otimes \sem{\Gamma_1}) \seq (\sem{B} \otimes \sem{\Gamma_2})
  \\ &
  \xlongrightarrow{\sem{N_1}\seq \sem{N_2}}
  \sem{C} \seq \sem{D}.
\end{align*}

For the term $\qifx{M}{N_1}{N_2}$,
to define the semantics, it suffices to show that
the morphism $\qifmap_A$ lifts to $\CausHilb$.
See \cref{app:sec:proof} for the proof.
\begin{proposition}
  \label{prop:qif-lifts-to-caushilb}
  For each object $\mathbf{A} \in \CausHilb$,
  the map $\qifmap_A$ in $\Hilb$ defines a map
  $\sem{\qbit} \otimes (\mathbf{A} \times \mathbf{A})
    \longrightarrow \sem{\qbit} \after \mathbf{A}$
  in $\CausHilb$.
  \lipicsEnd
\end{proposition}

\subsection{Full Abstraction}
\label{subsec:full-abst}

Though we do not provide an operational semantics for \ourlang{},
we can still discuss full abstraction for our language.
The following term $H_d$ is critical in the proof.

\begin{lemma}
  For each $d \in \Nat^+$,
  there is a program $\vdash H_d \colon d \otimes d$
  whose semantics is
  $\frac{1}{\sqrt d}\eta_d$
  where $\eta_d$ is
  the unit $\eta_d \colon I \longrightarrow d \otimes d$
  of the compact closed structure in $\Hilb$.
  \lipicsEnd
\end{lemma}

\begin{lemma}
  \label{lem:full-abst-inj-surg}
  For any type $A$,
  there are some natural numbers $n_A$, $m_A$, $\ell_A$
  and two terms
  $x_A \colon A \vdash M_A \colon n_A$
  and
  $y_A \colon m_A \vdash N_A \colon A \seq \ell_A$
  such that
  $\sem{M_A}$ is injective in $\Hilb$
  and
  $(\id_A \otimes \varepsilon_\ell) (\sem{N} \otimes \id_\ell)
    \colon m \otimes \ell \longrightarrow A$ is surjective in $\Hilb$.
\end{lemma}
\begin{proof}[Sketch of Proof]
  Ignoring some variable renamings and type equivalences,
  the terms can be recursively defined as follows:
  \begin{itemize}
    \item $M_d = M_\bot = x$,\ $N_d = N_\bot = y$,
    \item $M_{A \otimes B} = M_A \otimes M_B$,\
          $N_{A \otimes B} = \sinterch(N_A \otimes N_B)$,
    \item $M_{A \seq B} = \letx{x_A \seq x_B}{x}{(M_A \seq M_B)}$,\\
          $N_{A \seq B} =
          \letx{(a \otimes b) \seq z}
          {\sinterch(N_A \otimes N_B)}
          {(a \seq b) \seq z}$,
    \item For the function type, we can define the terms as following.
          Note that there is a equivalence of types
          $A \li ((B \seq \ell_B) \otimes n_A) \equiv
          (A \li B) \seq (\ell_B \otimes n_A)$
          via the term $\pexpose$.
          \begin{align*}
            M_{A \li B}
            =\ &
            \Let {y_A \otimes z^{m_A}} = {H_{m_A}}
            \In
            \Let {z^A \seq z^{\ell_A}} = N_A
            \In \\ &
            ((
              \Let x_B =  {x\,z^A} \In
              {M_B} 
            )
              \seq z^{\ell_B}
            )
            \otimes z^{m_A},
            \\
            N_{A\li B}
            =\ &
            \lambda x_A.\, N_B \otimes M_A.
            \tag*{\qedhere}
          \end{align*}
  \end{itemize}
\end{proof}

\begin{theorem}[Full abstraction]
  Let $\Gamma \vdash M \colon A$ and
  $\Gamma \vdash N \colon A$.
  Then, $\sem{C[M]} = \sem{C[N]}$
  for any context $C[-]$ of type $d \in \Nat^+$,
  if and only if
  $\sem{M} = \sem{N}$.
\end{theorem}
\begin{proof}
  From the compositionality of the semantics, the if part is trivial.
  We prove the converse.
  Let $\Gamma = \vec{x_i} \colon \vec{B_i}$.
  Then, using $N_{B_i}$ and $\sinterch$,
  we obtain a term $L$ of type
  $\vec{y_i} \colon \vec{m_{B_i}}
  \vdash L \colon (\bigotimes B_i) \seq (\prod \ell_{B_i})$.
  Now, we define a context $D[-]$ as follows:
  \begin{align*}
    \vec{y_i} \colon \vec{m_{B_i}}
    \vdash
    \Let {\vec{x_i} \seq z} = {L} \In
    (\Let x_A = [-] \In M_A) \seq z
    \colon n_A \seq \textstyle\prod \ell_{B_i}.
  \end{align*}
  From \cref{lem:full-abst-inj-surg},
  $\sem{D[-]}$ defines an injective map
  $\Hilb(\sem{\Gamma}, \sem{A})
  \longrightarrow
  \Hilb(\prod_i m_{B_i}, n_A \times \prod_i \ell_{B_i})$.
  For any map $f, g \in \Hilb(n, m)$,
  if $f \neq g$, there exists a vector $\ket\psi$ of length 1
  such that $f \ket\psi \neq g \ket\psi$.
  Therefore, if $\sem{M} \neq \sem{N}$,
  there exists some unitary map
  $\gU \colon \prod m_{B_i} \longrightarrow \prod m_{B_i}$
  that makes the context
  $
    C[-]\,\defeq\,\letx{\vec{y_i}}{{\gU \ket0}}{D[-]}
  $
  satisfy $\sem{C[M]} \neq \sem{C[N]}$.
\end{proof}

\subsection{Causality and Definability}
\label{subsec:definability}

In this section, we study the class of maps and supermaps
that can be described in \ourlang{}.
In particular, we prove that every term is causally consistent 
in the sense that no paradox described in the Introduction
occurs in \ourlang{}.
We do this by studying the definability
of the semantics taken in $\CausHilb$.

First, for first-order maps, 
we can prove the following:

\begin{proposition}\label{prop:semantic-soundness-of-the-type-system}
  The semantics of every term $x \colon n \vdash M \colon m$
  is an isometry.
  Conversely, every map in $\Isom(n, m)$ is definable
  by a term $x \colon n \vdash M \colon m$.
\end{proposition}
\begin{proof}
  It follows from $\CausHilb(n, m) = \Isom(n, m)$
  proved in \cref{thm:isom-is-full-subcategory}.
  For the definability,
  we can define each isometry by the term
  $\lambda x.\,\gU(\qinj(x))$.
\end{proof}

In particular, this fact shows that our language can
avoid the paradox regarding the quantum conditionals
we presented in the Introduction.

We can also characterize all pure supermaps:

\begin{theorem}\label{thm:characterization}
  A supermap is pure if and only if
  it can be represented as $\iota f$ for some
  $f \colon \bigotimes_i(\mathbf{n_i} \li \mathbf{m_i}) \li \mathbf{n} \li \mathbf{m}$
  in $\CausHilb$.
  In particular, all definable supermaps
  $\bigotimes_i(n_i \li m_i) \li n \li m$ in our language are pure,
  and the OCB process~\cite{Oreshkov2012-ocb}
  is not definable in \ourlang{} since it is not pure.
\end{theorem}
\begin{proof}
  Sufficiency follows from \cref{thm:isom-is-full-subcategory}
  and a similar observation as in \cref{ex:auxiliary-space}.
  We prove the necessity.
  A similar proof can be found in \cite{Yokojima2021consequencesof}.
  For each pure supermap
  $g \colon \bigotimes_i(\Lin(\CC^{2^{n_i}}) \li \Lin(\CC^{2^{m_i}}))
  \longrightarrow \Lin(\CC^{2^{n}}) \li \Lin(\CC^{2^{m}})$,
  by partially applying
  $\mathsf{swap}_i\colon n_i \otimes m_i \longrightarrow m_i \otimes n_i$,
  we obtain an isometry of type
  $n \otimes_i m_i \longrightarrow m \otimes_i n_i$.
  Since the original $g$ can be recovered by composing unit and counit of
  the compact closed structure, and all these maps are in the image of the functor $\iota \colon \Hilb \longrightarrow \CPM$,
  the map $g$ itself is in the image of $\iota$.
  Let $\iota f = g$, then since $f\in\Hilb$ maps isometries to isometry,
  it is a map in $\CausHilb$.
\end{proof}

Furthermore, using some facts that have been proven in the quantum community,
we can extend the result to the supermaps
that take at most two inputs:

\begin{proposition}
  Every map of type
  $(\mathbf{n} \li \mathbf{m}) \li (\mathbf{n'} \li \mathbf{m'})$
  in $\CausHilb$
  is definable.
\end{proposition}
\begin{proof}
  For single input supermaps (deterministic supermaps),
  in \cite[Theorem 1]{Chiribella_2008},
  it is shown that all such maps can be written as
  $\lambda f.\, \gV \circ (f \otimes \id)\circ \gU \circ \qinj$.
\end{proof}

To discuss supermaps with two inputs, we extend our language
by extending the syntax of $\qif$ to
$\qifx{M \geq m}{x \to N_1}{y \to N_0}$.
This represents a coherent branching where,
it branches to $N_0$
if the control qubit $M$ is $\ket{0}$, \dots, $\ket{m-1}$,
and it branches to $N_1$
if it is in between $\ket{m}$, \dots $\ket{n + m - 1}$.
The typing rule is now modified as follows.
\begin{proofrule}
  \infer*[]{
    \Gamma \vdash M \colon n + m
    \\
    \Gamma', x_0 \colon m \vdash N_0 \colon m \after A
    \\
    \Gamma', x_1 \colon n \vdash N_1 \colon n \after A
  }{
    \Gamma, \Gamma' \vdash
    \qifx{M \geq m}{x_1 \to N_1}{x_1 \to N_0} \colon (n + m) \after A
  }
\end{proofrule}
The original $\qifx{M}{N_1}{N_2}$ can be recovered
by $\qifx{M \geq 1}{x \to (x; N_1)}{y \to (y; N_0)}$
when $m = n = 1$.
The semantics is defined in \cref{app:sec:proof}.

\begin{proposition}
  \label{prop:two-shot}
  In the language extended by generalizing
  $\qif$ that takes a qu$d$it as its condition,
  every map $(\mathbf{n} \li \mathbf{m}) \li (\mathbf{n} \li \mathbf{m}) \li (\mathbf{n'} \li \mathbf{m'})$
  in $\CausHilb$ is definable.
\end{proposition}
\begin{proof}
  For supermaps with two inputs (bipartite supermaps),
  in \cite[Theorem 5]{Yokojima2021consequencesof},
  it was shown that such maps can be written by a direct sum of
  two causally ordered processes.
  That is, there is a decomposition
  of the supermap $f$ 
  \[
    f(g,h) =
    V \circ (f_1(g,h) \oplus f_2(g,h)) \circ U \circ \qinj
  \]
  for some unitaries $U$, $V$ and quantum combs
  $f_i \colon (n \li m) \li (n \li m) \li n_i \li m_i$
  which can be realized by a composition of isometries.
  Since this $\oplus$ can be represented
  with a generalized $\qif$,
  this is definable.
\end{proof}

\begin{remark}
  However, we conjecture that for supermaps that take three or more inputs,
  there are some maps which are not definable in our language.
  This is because it is known that there are some tripartite supermaps
  that break causal inequality~\cite{AraujoFNB17},
  which do not seem to be describable only with $\qif$.
  \lipicsEnd
\end{remark}

\section{Conclusion}
\label{sec:conclusion}
We have developed a higher-order quantum programming language equipped with a type system based on causal logic, together with its categorical semantics.
In the presence of quantum conditional branching controlling functions, a na\"ive linear type discipline cannot avoid some physically-unrealizable programs.
In this paper, we have shown that tools from pomset/BV/causal logics~\cite{retore_pomset_1997,Guglielmi2007,Simmons2024}, namely the seq-connective \(\seq\) and the concept of first-order types, resolve this issue.
While it has been known that semantic models of quantum computation can exhibit structures of causality-aware logics, our results strengthen this connection by showing that some ideas from pomset/BV/causal logics are in fact indispensable in higher-order quantum computation.

\bibliographystyle{plainurl}%
\bibliography{references}

\clearpage
\appendix
\section{IBV, SequentIBV, and Our Language}
\label[appendix]{app:sec:logic-lang}

\subsection{Intuitionistic BV-logic}

For the derivation rules of IBV, see \cite{AcclavioS25}.

\begin{proof}[Proof of \cref{prop:IBV-eq-sequentIBV}]
  Proving that
  all derivable propositions in sequentIBV are derivable in IBV
  is easy.
  In \cref{app:fig:ibv-basic-derivable-in-sequentIBV}
  we describe how an IBV basic proposition can be derived in sequentIBV.
  In \cref{app:fig:ibv-deepinference-derivable-in-sequentIBV},
  we show how a deep inference can be simulated in a sequent calculus.
\end{proof}
\begin{figure*}
  \begin{proofrules}[0.20cm]
    \infer*{
      \infer*{
        A \vdash A
      }{
        \vdash A \li A
      }
    }{
      I \vdash A \li A
    }

    \infer*{
      \vdash I \seq I
      \\
      \infer*{
        I \vdash I
        \\
        \infer*{
          A \vdash A
        }{
          A, I \vdash A
        }
      }{
        A, I \seq I \vdash I \seq A
      }
    }{
      A \vdash I \seq A
    }

    \infer*{
      \vdash I \seq I
      \\
      \infer*{
        \infer*{
          A \vdash A
        }{
          A, I \vdash A
        }
        \\
        I \vdash I
      }{
        A, I \seq I \vdash A \seq I
      }
    }{
      A \vdash A \seq I
    }

    \infer*{
      \infer*{
        \vdash I \seq I
        \\
        \infer*{
          \infer*{A\vdash A}{A, I\vdash A}
          \\
          \infer*{B\vdash B}{B, I\vdash B}
        }{
          A, B, I \seq I \vdash A \seq B
        }
      }{
        A, B \vdash A \seq B
      }}{
      A \otimes B \vdash A \seq B
    }

    \infer*{
      \infer*{
        \infer*{
          A \vdash A
          \\
          \infer*{
            B \vdash B
            \\
            C \vdash C
          }{
            B, B \li C \vdash C
          }
        }{
          A, A \li B, B \li C  \vdash C
        }
      }{
        A, B \li C \vdash (A \li B) \li C
      }
    }{
      A \otimes (B \li C) \vdash (A \li B) \li C
    }

    \infer*{
      \infer*{
        \infer*{
          A \vdash A
          \\
          \infer*{
            B \vdash B
            \\
            C \vdash C
          }{
            B, C \vdash B \otimes C
          }
        }{
          A, A \li B, C \vdash B \otimes C
        }
      }{
        A \li B, C \vdash A \li (B \otimes C)
      }
    }{
      (A \li B) \otimes C \vdash A \li (B \otimes C)
    }

    \infer*{
      \infer*{
        \infer*{
          A \vdash A
          \\
          B \vdash B
        }{
          A, A \li B \vdash B
        }
        \\
        C \vdash C
      }{
        A, (A \li B) \seq C \vdash B \seq C
      }
    }{
      (A \li B) \seq C
      \vdash
      A \li (B \seq C)
    }

    \infer*{
      \infer*{
        B \vdash B\\
        \infer*{
          A \vdash A
          \\
          C \vdash C
        }{
          A, A \li C \vdash C
        }
      }{
        A, B \seq (A \li C) \vdash B \seq C
      }
    }{
      B \seq (A \li C) \vdash A \li (B \seq C)
    }

    \infer*{
      \infer*{
        A\seq C
        \vdash
        A\seq C
        \\
        B \seq D
        \vdash
        B \seq D
      }{
        A\seq C, B \seq D
        \vdash
        (A\otimes B) \seq (C \otimes D)
      }
    }{
      (A\seq C) \otimes (B \seq D)
      \vdash
      (A\otimes B) \seq (C \otimes D)
    }

    \infer*{
      (A\li B) \seq (C \li D)
      \vdash
      (A\li B) \seq (C \li D)
    }{
      (A\li B) \seq (C \li D)
      \vdash
      (A\seq C) \li (B \seq D)
    }

      \infer*{
        B \vdash B\\
        A \vdash A
      }{
        A \otimes B \vdash B \otimes A
      }

    \infer*{
      \infer*{
        \infer*{
          A \vdash A
          \\
          \infer*{
            B \vdash B
            \\
            C \vdash C
          }{
            B, C
            \vdash
            B \otimes C
          }
        }{
          A, B, C
          \vdash
          A \otimes (B \otimes C)
        }
      }{
        A \otimes B, C
        \vdash
        A \otimes (B \otimes C)
      }
    }{
      (A \otimes B) \otimes C
      \vdash
      A \otimes (B \otimes C)
    }

    \infer*{
      \infer*{
        \infer*{
          \infer*{
            A \vdash A\\
            B \vdash B
          }{
            A, B \vdash A \otimes B
          }
          \\
          C \vdash C
        }{
          A, B, (A \otimes B) \li C
          \vdash
          C
        }
      }{
        A, (A \otimes B) \li C
        \vdash
        B \li C
      }
    }{
      (A \otimes B) \li C
      \vdash
      A \li B \li C
    }

    \infer*{
      \infer*{
        \infer*{
          A \vdash A\\
          \infer*{
            B \vdash B\\
            C \vdash C
          }{
            B, B \li C
            \vdash C
          }
        }{
          A, B,
          A \li B \li C
          \vdash
          C
        }
      }{
        A \otimes B,
        A \li B \li C
        \vdash
        C
      }
    }{
      A \li B \li C
      \vdash
      (A \otimes B) \li C
    }

    \infer*{
      \infer*{
        \infer*{
          A \vdash A \\ B \vdash B
        }{
          A, B \vdash A \otimes B
        }
        \\
        C \vdash C
      }{
        A, B \seq C
        \vdash (A \otimes B) \seq C
      }
    }{
      A \otimes (B \seq C)
      \vdash (A \otimes B) \seq C
    }

    \infer*{
      \infer*{
        B \vdash B
        \\
        \infer*{
          A \vdash A \\ C \vdash C
        }{
          A, C \vdash A \otimes C
        }
      }{
        A, B \seq C
        \vdash B \seq (A \otimes C)
      }
    }{
      A \otimes (B \seq C)
      \vdash B \seq (A \otimes C)
    }

    \infer*{
      \vdash I\\
      A \vdash A
    }{
      A \vdash I \otimes A
    }

    \infer*{
      \infer*{
        A \vdash A
      }{
        I, A \vdash A
      }
    }{
      I \otimes A \vdash A
    }

    \infer*{
      \infer*{
        A \vdash A
      }{
        A, I \vdash A
      }
    }{
      A \vdash I \li A
    }

    \infer*{
      \vdash I
      \\
      A \vdash A
    }{
      I \li A \vdash A
    }

    (A \seq B) \seq C
    \dashvdash
    A \seq (B \seq C)
  \end{proofrules}
  \caption{Basic inference rules in IBV are derivable in sequentIBV.}
  \label{app:fig:ibv-basic-derivable-in-sequentIBV}
\end{figure*}

\begin{figure*}
  \begin{proofrules*}
    \infer*{
      \infer*{
        \infer*{\pi}{
          X \vdash Y
        }
        \\
        A \vdash A
      }{
        X, A \vdash Y \otimes A
      }
    }{
      X \otimes A \vdash Y \otimes A
    }

    \infer*{
      \infer*{
        A \vdash A
        \\
        \infer*{\pi}{
          X \vdash Y
        }
      }{
        A, A \li X \vdash Y
      }
    }{
      A \li X \vdash A \li Y
    }

    \infer*{
      \infer*{
        \infer*{\pi}{
          X \vdash Y
        }
        \\
        A \vdash A
      }{
        X, Y \li A \vdash A
      }
    }{
      Y \li A \vdash X \li A
    }

    \infer*{
      \infer*{\pi}{
        X \vdash Y
      }
      \\
      A \vdash A
    }{
      X \seq A \vdash Y \seq A
    }

    \infer*{
      A \vdash A
      \\
      \infer*{\pi}{
        X \vdash Y
      }
    }{
      A \seq X \vdash A \seq Y
    }
  \end{proofrules*}
  \caption{sequentIBV derivations corresponding to deep inference rules.}
  \label{app:fig:ibv-deepinference-derivable-in-sequentIBV}

\end{figure*}

\begin{proof}[Proof of \cref{prop:program-as-proof}]
  For each derivation $A_1, \dots, A_n \vdash B$ in sequentIBV,
  a program
  $x_1 \colon A_1, \dots, x_n \colon A_n \vdash M \colon B$
  can be defined inductively as in \cref{app:fig:ibv-proof-as-program}.
\end{proof}
\begin{figure*}
  \begin{proofrules*}
    \infer
    {}
    {x \colon A \vdash x \colon A}

    \infer
    {\Gamma, x \colon A, y \colon B \vdash M \colon C}
    {\Gamma, z \colon A \otimes B \vdash \letx{x \otimes y}{z}{M} \colon C}

    \infer
    {\Gamma \vdash M \colon A\\
    \Gamma' \vdash N \colon B}
    {\Gamma, \Gamma' \vdash M \otimes N \colon A \otimes B}

    \infer
    {\Gamma, x \colon A \vdash M \colon B}
    {\Gamma \vdash \lambda x.\,M \colon A \li B}

    \infer
    {\Gamma \vdash M \colon A\\
    \Gamma', x \colon B \vdash N \colon C}
    {\Gamma, \Gamma', f \colon A \li B \vdash \letx{x}{f\,M}{N} \colon C}

    \infer{\Gamma \vdash M \colon A}
    {\Gamma, x \colon \bot \vdash \letx{\ \_}{x}{M} \colon A}

    \infer{}
    {\vdash \unitval \colon \bot}

    \infer
    {\Gamma \vdash M \colon A\\
    \Gamma', x \colon A \vdash N \colon B}
    {\Gamma, \Gamma' \vdash \letx{x}{M}{N} \colon B}

    \infer
    {}
    {\ \  \vdash \unitval \seq \unitval \colon I \seq I}

    \infer
    { z \colon A \seq C \vdash z \colon A \seq C\\ 
      \Gamma, x \colon A \vdash M \colon B\\
    \Gamma', y \colon C \vdash N \colon D}
    {\Gamma, \Gamma', z \colon A \seq C \vdash
    \letx{x \seq y}{z}{(M \seq N)} \colon B \seq D}

    \infer
    {
      \Gamma \vdash M \colon A \seq C\\
      \Gamma' \vdash N \colon B \seq D
    }{
      \Gamma, \Gamma' \vdash \sinterch(M \otimes N)
      \colon (A \otimes B) \seq (C \otimes D)
    }

    \infer*
    {
      \infer*{
      \infer*{
        \Gamma \vdash M \colon (A \li C) \seq (B \li D)
      }{
        \Gamma, x \colon A \seq B \vdash M \otimes x \colon 
        ((A \li C) \seq (B \li D)) \otimes (A \seq B)
      }
      }{
        \Gamma, x \colon A \seq B \vdash \sinterch(M \otimes x) \colon 
        ((A \li C) \otimes A)  \seq ((B \li D) \otimes B)
      }
    }{
      \Gamma \vdash
      \lambda x.
      \letx{y \seq z}{
        (\sinterch{M \otimes x})
      }{
          (\letx{y_1 \otimes y_2}{y}{y_1\, y_2})
          \seq
          (\letx{z_1 \otimes z_2}{z}{z_1\, z_2})
      }
      \\
      \colon 
      (A \seq B) \li (C \seq D)
    }

    \infer{}{
      x \colon (A \seq B) \seq C
      \vdash
      \sassocL(x) \colon A \seq (B \seq C)
    }

    \infer{}{
      x \colon A \seq (B \seq C)
      \vdash
      \sassocR(x) \colon (A \seq B) \seq C
    }
  \end{proofrules*}
  \caption{Programs in \ourlang{} corresponding to proofs in sequentIBV.}
  \label{app:fig:ibv-proof-as-program}
\end{figure*}

\section{Detailed proof of existence of Cartesian products}
\label[appendix]{app:sec:product}

\begin{proof}[Proof of \cref{prop:caushilb-product}]
  Since each $\mathbf{A} = (n, c_\mathbf{A})$ is is isomorphic to
  $(n, \{\lambda\rho\mid\rho\in c_\mathbf{A}\})$
  for each positive real $\lambda$,
  without loss of generality,
  we can assume that $\maxmixDiag_n \in c_\mathbf{A}$
  and $\frac1n\cdot\discardDiag_n\in c_\mathbf{A}$.
  We prove that
  \[
    c_{\mathbf{A}\times\mathbf{B}}
    \defeq
    \{
    \rho \mid (\iota p_A)\rho \in c_\mathbf{A}, (\iota p_B)\rho \in c_\mathbf{B}
    \}
  \]
  is flat for each
  $\mathbf{A} = (n, c_\mathbf{A})$
  and
  $\mathbf{B} = (m, c_\mathbf{B})$.
  First,
  $\maxmixDiag_{n + m} \in c_{\mathbf{A}\times\mathbf{B}}$
  because $(\iota p_A) \circ \maxmixDiag_{n + m} = \maxmixDiag_n$
  and
  $(\iota p_B) \circ \maxmixDiag_{n + m} = \maxmixDiag_m$.
  Also,
  $\frac{1}{n + m}\discardDiag_{n+m} \in c_{\mathbf{A}\times\mathbf{B}}^*$,
  because
  for each $\rho \in c_{\mathbf{A}\times\mathbf{B}}$,
  \begin{align*}
    \left(
      \frac{1}{n + m} \discardDiag_{n+m}
    \right)
      \circ \rho
    &
    = \frac{1}{n + m} \sum_{i = 1}^{n+m}
    \langle i | \rho | i \rangle
    = \frac{1}{n + m}
    \left(
    \sum_{k = 1}^{n} \langle k | \rho | k \rangle
    +
    \sum_{\ell = 1}^{m} \langle \ell | \rho | \ell \rangle
    \right)
    \\ &
    = \frac{1}{n + m}
    \big(\discardDiag_n \circ\,(\iota p_A) \circ \rho
    +
    \discardDiag_m \circ (\iota p_B) \circ \rho\big)
    = \frac{1}{n + m} \cdot (n + m)
    = 1.
  \end{align*}

  For the closedness of $c_{\mathbf{A}\times\mathbf{B}}$,
  from \cite[Theorem 17]{Simmons2022},
  it suffices to prove that, for each sequence of reals
  $(a_i)_{i = 0,\dots,k}$ such that
  $\sum_i a_i = 1$,
  if $\rho = \sum_i a_i\rho_i \in \CPM(I, (n + m))$,
  then $\rho \in c_{\mathbf{A}\times\mathbf{B}}$.
  This is clear from the definition.

  To complete the proof,
  we need to prove that the projections and universal maps
  are causal.
  The causality of projections $p_\mathbf{A}$ and $p_\mathbf{B}$
  is clear from the definition.
  To prove the causality of universality,
  let us assume that we are given two maps
  $f \in \CausHilb(\mathbf{X}, \mathbf{A})$
  and
  $g \in \CausHilb(\mathbf{X}, \mathbf{B})$.
  Then it suffices to prove that,
  for each $\rho \in c_{\mathbf{X}}$,
  $(\iota \langle f,g\rangle) \circ \rho$
  is included in
  $c_{\mathbf{A} \times \mathbf{B}}$.
  This holds because
  $(\iota p_\mathbf{A}) \circ
  (\iota \langle f, g\rangle) \circ
  \rho
  = (\iota f) \circ \rho
  \in c_\mathbf{A}$.
\end{proof}

\section{Proofs in \cref{sec:cat-sem-causHilb}}
\label[appendix]{app:sec:proof}

\begin{proof}[Proof of \cref{prop:qif-lifts-to-caushilb}]
  Let $\mathbf{2} \defeq \sem{\qbit} = (2, \{\discardDiag\}^*)$.
  To prove the claim,
  we use results shown in \cite{Simmons2022}.
  First, to prove that $\qifmap_A$ has a type
  $\mathbf{2} \otimes (\mathbf{A} \times \mathbf{A})
    \li \mathbf{2} \after \mathbf{A}
    \iso \mathbf{2} \llpar \mathbf{A}$,
  it suffices to prove the following:
  \begin{align*}
    \forall \rho_\mathbf{2} \in c_\mathbf{A}.\
    \forall \rho_{\mathbf{A} \times \mathbf{A}}
    \in c_{\mathbf{A} \times \mathbf{A}}.\
    \forall \pi \in c^*_\mathbf{2}.
    \qquad(\pi \otimes \id_{\EmbeddingFunctor A})
    \circ\EmbeddingFunctor(\qifmap_A)\circ
    (\rho_\mathbf{2} \otimes \rho_{\mathbf{A} \times \mathbf{A}})
    \in c_\mathbf{2}.
  \end{align*}
  Since $\mathbf{2}$ is first-order,
  $c_\mathbf{2}^* = \{ \discardDiag \}$.
  Also, since every density matrix $\rho_\mathbf{2}$
  can be written as a convex combination of pure states,
  and since $c_\mathbf{2}$ is closed under convex combination,
  it suffices to assume $\rho_\mathbf{2}$ is pure.
  Let $\ket\psi = \alpha \ket0 + \beta \ket1$ be
  a quantum state such that $|\alpha|^2 + |\beta|^2 = 1$.
  We prove the following for every
  $\rho_{\mathbf{A} \times \mathbf{A}} \in c_{\mathbf{A} \times \mathbf{A}}$.
  \begin{align*}
    (\discardDiag \otimes \id_{\EmbeddingFunctor A})
    \circ\EmbeddingFunctor(\qifmap_A)\circ
    (\ket\psi\!\bra\psi \otimes \rho_{\mathbf{A} \times \mathbf{A}})
    \in c_\mathbf{2}
  \end{align*}
  Since the control qubit $\ket\psi\!\bra\psi$ of
  $\qifmap_A$ is discarded right after the $\qifmap$,
  this is equivalent to the following:
  \begin{align*}
    |\alpha|^2 \iota(p)(\rho_{\mathbf{A} \times \mathbf{A}})
    +
    |\beta|^2 \iota(p')(\rho_{\mathbf{A} \times \mathbf{A}})
  \end{align*}
  where $p$ and $p'$ are the projections
  $A \xleftarrow{p} A \times A \xrightarrow{p'} A$ in $\Hilb$.
  By the definition of
  $c_{\mathbf{A} \times \mathbf{A}}$,
  $\iota(p)(\rho_{\mathbf{A} \times \mathbf{A}})$
  and
  $\iota(p')(\rho_{\mathbf{A} \times \mathbf{A}})$
  are elements of $c_\mathbf{A}$.
  Since $|\alpha|^2 + |\beta|^2 = 1$,
  this is a convex combination of elements of $c_\mathbf{A}$,
  so it is also an element of $c_\mathbf{A}$.
\end{proof}

\begin{proof}[Proof of \cref{lem:full-abst-inj-surg}]
  Here, we provide a slightly more detailed proof
  that was sketched in \cref{subsec:full-abst}.
  Let us denote the term $M_A$ and $N_A$ as
  \begin{align*}
    x_A \colon A    \vdash M_A \colon n_A,
    \qquad
    y_A \colon m_A  \vdash N_A \colon A \seq \ell_A.
  \end{align*}
  \begin{itemize}
    \item For the base case when $A \iso d \in \Nat^+$,
          we define $M_A \defeq x_A$ and $N_A \defeq y_A$ and
          $n_A = m_A = d$ and $\ell_A = 0$.
          The semantics of the term $M_A$ or $N_A$ is
          the identity, thus the claim holds.
    \item For the tensor type,
          $M_{A\otimes B}$ is defined by
          $\letx{x_A \otimes x_B}{x_{A \otimes B}}{(M_A \otimes M_B)}$,
          and $N_{A \otimes B}$ is defined by
          $\letx{y_A \otimes y_B}{y_{A \otimes B}}{\sinterch(N_A \otimes N_B)}$.
          Since the tensor product of an injection (resp. surjection)
          is also an injection (resp. surjection) in $\Hilb$, the claim holds.
    \item For the seq type,
          $M_{A \seq B}$ is defined by
          $\pawait(
            \letx{x_A \seq x_B}{x_{A \seq B}}{(M_A \seq M_B)}
            )$.
          So the underlying morphism of the semantics
          is isomorphic to $\sem{M_A} \otimes \sem{M_B}$, which is injective.

          Also, the term $N_{A \seq B}$ is defined by
          \begin{align*}
            N_{A \seq B}
            \defeq
            \letx{c \seq z}
            {\sinterch(N_A \otimes N_B)}
            {(\Let {a \otimes b} = {c} \In {(a \seq b)}) \seq z}
          \end{align*}
          whose underlying morphisms are isomorphic to
          $\sem{N_A} \otimes \sem{N_B}$.
    \item We define $M'_{A \li B}$ as:
          \begin{align*}
            \letx{y_A \otimes z^{m_A}}{H_{m_A}}{
            \letx{z^A \seq z^{\ell_A}}{N_A}{
            (( \letx{x_B}{x\,z^A}{M_B}) \seq z^{\ell_B})
            \otimes z^{m_A}
            }}
          \end{align*}
          whose type is
          $(n_B \seq \ell_B) \otimes m_A$.
          This type is isomorphic to
          $n_B \otimes \ell_B \otimes m_A$,
          so we obtain a term $M_{A\li B}$ of that type.
          The semantics of this term is
          a map that maps $f \in \Hilb(A, B)$ to
          \begin{align*}
            I
             & \xlongrightarrow{\tfrac{1}{\sqrt{m_A}}\eta_{(m_A)}}
            m_A \otimes m_A
            \xlongrightarrow{\sem{N_A} \otimes \id}
            A \otimes \ell_A \otimes m_A
            \\
             & \xlongrightarrow{f \otimes \id}
            B \otimes \ell_A \otimes m_A
            \xlongrightarrow{\sem{M_B} \otimes \id}
            n_B \otimes \ell_A \otimes m_A.
          \end{align*}
          This procedure can be decomposed to
          (1) apply a faithful functor $(-) \otimes \id_{\ell_A \otimes m_A}$,
          (2) pre-compose a surjection,
          (3) post-compose an injection,
          (4) bend some wires and multiply a non-zero scalar.
          Therefore, it is injective.
    \item We define $N_{A\li B}'$ by
          \begin{align*}
            y_{B} \colon m_B
            \vdash
            \lambda x_A.\, N_B \otimes M_A \colon
            A \li ((B \seq \ell_B) \otimes n_A).
          \end{align*}
          Since $\ell_B$ and $n_A$ are first-order,
          we can define $N_{A \li B}$ of type
          $(A \li B) \seq (\ell_B \otimes n_A)$ using $\pexpose$.
          The linear map
          $f \defeq (\id \otimes \varepsilon_{\ell_B \times n_A})
            \circ (\sem{N_{A\li B}} \otimes \id)$
          maps
          $\sum_i \ket{a_i}\ket{b_i}\ket{c_i}
            \in \CC^{m_B} \otimes \CC^{\ell_B} \otimes \CC^{n_A}$
          to the function
          \begin{align*}
            \ket{\psi} \in \CC^{\sem{A}}
            \longmapsto
            \sum_i
            \bra{b_i} \sem{M_A} \ket\psi \otimes
            (\id_{\sem{B}} \otimes \bra{c_i}) \sem{N_B} \ket{a_i}.
          \end{align*}
          For any linear map
          $g \colon \ket\psi \longmapsto
            \sum_{j} \langle{p_j}|{\psi}\rangle \otimes \ket{q_j}$,
          we can choose some $\sum_i \ket{a_i}\ket{b_i}\ket{c_i}$
          that is mapped to $g$ by $f$
          since $\sem{M_A}$ is injective and
          $\sum_i\ket{a_i}\ket{c_i} \longmapsto
            (\id_{\sem{B}} \otimes \bra{c_i}) (\sem{N_B} \ket{a_i})$
          is surjective.
          Therefore $f$ is surjective.
          \qedhere
  \end{itemize}
\end{proof}

\subsection*{Generalized qif for qudits}

We define the semantics of generalized $\qif$
in $\Hilb$ simply via the direct sum of maps:
\begin{align*}
  \sem{\qifx{M \geq n}{x\to N_1}{y\to N_0}}_\Hilb
  \colon\hspace{-10em}
  \\
  \sem{\Gamma} \otimes \sem{\Gamma'}
   &
  \xrightarrow{\sem{M}\otimes \id}
  (n \oplus_\Hilb m) \otimes \sem{\Gamma'}
  \ \cong\
  (\sem{\Gamma'} \otimes m)
  \oplus_\Hilb
  (\sem{\Gamma'} \otimes n)
  \\
   &
  \xrightarrow{\sem{N_0} \oplus \sem{N_1}}
  (m \otimes \sem{A})
  \oplus_\Hilb
  (n \otimes \sem{A})
  \ \cong \
  (n \oplus_\Hilb m) \otimes \sem{A}
  .
\end{align*}

Now we need to check that this map lifts to
a map in $\CausHilb$.
To this end, it suffices to check the following
$\qifmap_A'$ defines a morphism.

\begin{proposition}
  Let us define a linear map $\qifmap_A'$ as follows:
  \begin{align*}
    \begin{matrix}
      (n \li (n \otimes A))
      \oplus_\Hilb
      (m \li (m \otimes A))
       &
      \longrightarrow
       &
      (n \oplus_\Hilb m)
      \li
      ((n \oplus_\Hilb m) \otimes A)
      \\
      (f,\ g)
       &
      \longmapsto
       &
      f \oplus g
    \end{matrix}
  \end{align*}
  This map $\qifmap_A'$ defines a map of type
  \begin{align*}
    (\mathbf{n} \li (\mathbf{n} \llpar \mathbf{A}))
    \times
    (\mathbf{m} \li (\mathbf{m} \llpar \mathbf{A}))
    \longrightarrow
    (\mathbf{n + m}) \li ((\mathbf{n + m}) \llpar \mathbf{A})
  \end{align*}
  in $\CausHilb$.\footnote{
    Note that $\mathbf{n + m}$ is the sum of natural numbers
    and not the Categorical (co)product.
  }
\end{proposition}
\begin{proof}
  Let $f_x \in \Hilb(n, n \otimes A)$
  and $g_x \in \Hilb(m, m \otimes A)$ be maps
  such that
  $\sum_{x \in X} (f_x, g_x)
    \in c_{(\mathbf{n} \li (\mathbf{n} \llpar \mathbf A))
        \times (\mathbf{n} \li (\mathbf{n} \llpar \mathbf A))}$.
  Let $\ket\varphi \in \CC^n$ and $\ket\psi \in \CC^m$
  be normalized vectors.
  We identify these vectors with vectors
  in the coproduct space $\CC^{n+m}$.
  Let $\alpha$ and $\beta$ be complex numbers such that
  $|\alpha|^2 + |\beta|^2 = 1$.

  Using the affine closedness of flat and closed sets~\cite{Simmons2022}
  and linearity of maps,
  without loss of generality,
  we can focus on pure inputs.
  That is, it suffices to prove that the following is in $c_\mathbf{A}$:
  \begin{align*}
    \textstyle
    (\discardDiag_{n+m} \otimes \id_A)
    \circ
    \iota(\sum_x f_x \oplus_\Hilb g_x)
    \circ
    \iota(\alpha \ket\varphi + \beta \ket\psi).
  \end{align*}
  Let $p \colon \mathbf{n}+\mathbf{m} \longrightarrow
    \mathbf{n}\oplus_\CPM \mathbf{m}$
  be the canonical map that partially measures the qu$d$its.
  We can calculate as follows.
  \begin{align*}
     & \textstyle
    =
    ((\discardDiag_{n} \oplus_\CPM \discardDiag_m)
    \circ p
    \otimes \id_A)
    \circ
    \iota(\sum_x f_x \oplus g_x)
    \circ
    \iota(\alpha \ket\varphi + \beta \ket\psi).
    \\
     & \textstyle
    =
    ((\discardDiag_{n} \oplus_\CPM \discardDiag_m)
    \circ p
    \otimes \id_A)
    \circ
    \iota(\sum_x \alpha \cdot f_x \ket\varphi
    \oplus_\Hilb \beta \cdot g_x \ket\psi)
    \\
     & \textstyle
    =
    ((\discardDiag_{n} \oplus_\CPM \discardDiag_m)
    \otimes \id_A)
    \circ
    (
      |\alpha|^2 \cdot \iota(\sum_x f_x \ket\varphi)
      \oplus_\CPM
      |\beta|^2 \cdot \iota(\sum_x g_x \ket\psi)
    )
    \\
     & \textstyle
    =
    |\alpha|^2 \cdot
    (\discardDiag_{n} \otimes \id_A)
    \circ
    \iota(\sum_x f_x \ket\varphi)
    \ \oplus_\CPM\
    |\beta|^2 \cdot
    (\discardDiag_{m} \otimes \id_A)
    \circ
    \iota(\sum_x g_x \ket\psi)
  \end{align*}
  By definition of product in $\CausHilb$,
  $\sum_x f_x \in \CausHilb(\mathbf{n}, \mathbf{n}\llpar\mathbf{A})$
  and
  $\sum_x g_x \in \CausHilb(\mathbf{m}, \mathbf{m}\llpar\mathbf{A})$.
  Therefore, this is a convex conbination
  of elements in $c_\mathbf{A}$,
  which is also an element of $c_\mathbf{A}$.
\end{proof}

\end{document}